\newif\ifijfcs
\ijfcsfalse
\ifijfcs
\documentclass{ws-ijfcs}
\usepackage{enumerate}
\usepackage{url}
\numberwithin{theorem}{section}
\newtheorem{conjecture}[theorem]{Conjecture}

\newtheorem{problem}[theorem]{Problem}

\newcommand{\qed}{}
\else
\documentclass[runningheads,envcountsame,envcountsect]{llncs}
\fi

\usepackage{graphicx}
\usepackage{wrapfig}
\usepackage{amsmath}
\usepackage{amssymb}
\ifijfcs
\else
\usepackage{paralist}
\fi	
\usepackage{tikz}

\newcommand{\N}{\mathbb{N}}
\newcommand{\Ns}{\mathbb{N}^{*}}
\newcommand{\bcdot}{$\discretionary{\mbox{$ \cdot $}}{}{}$}

\definecolor{Pepijn}{rgb}{0,0.5,1}
\definecolor{Michiel}{rgb}{1,0.7,0}
\definecolor{Twins}{rgb}{0.1,0.5,0.0} 
\definecolor{lightTwins}{rgb}{0.6,1,0.5}

\colorlet{lightPepijn}{Pepijn!50}
\colorlet{lightMichiel}{Michiel!50}
\colorlet{darkTwins}{Twins!50!black}

\colorlet{tableIndex}{lightTwins}
\colorlet{tableBold}{lightMichiel}
\colorlet{tableVoid}{lightPepijn}
\colorlet{tableFrame}{Twins!50!gray}

\colorlet{Chaser}{red}
\colorlet{Resigner}{Michiel}
\colorlet{lastPawn}{Twins}
\colorlet{pawnCircle}{darkTwins}

\colorlet{CernyColA}{white}
\colorlet{CernyColB}{lightMichiel}
\colorlet{CernyColC}{lightTwins}

\tikzset{
	myline/.style={lastPawn,line width=0.8mm}, 
	mycirc/.style={pawnCircle,line width=0.4mm,radius=1.2mm}, 
	mydash/.style={dash pattern=on 2pt off 2pt}, 
	mynode/.style={circle,inner sep=0pt,minimum width=5mm} 
}

\newcommand{\Cerny}[4][0]{\begin{scope}[semithick]
		\pgfmathsetmacro{\g}{#1} 
		\pgfmathsetmacro{\n}{#2} 
		\pgfmathsetmacro{\c}{#3} 
		\pgfmathsetmacro{\r}{#4} 
		\pgfmathsetmacro{\nmm}{\n-1}
		\pgfmathsetmacro{\nmc}{\n-\c}
		\pgfmathsetmacro{\nmcmm}{\n-\c-1}
		\foreach \i in {1,...,\n}{
			\pgfmathsetmacro{\angle}{90-0.5*\g-(360-\g)*(\i-0.5)/\n}
			\pgfmathsetmacro{\CernyCol}{\i<\n-\c?"CernyColA":(\i<\n?"CernyColB":"CernyColC")}
			\node[mynode,draw,fill=\CernyCol] (\i) at (\angle:\r) {{$\i$}};
		}
		\ifx2#2 
		\pgfmathsetmacro{\h}{(\r+0.1)/3} 
		\pgfmathsetmacro{\ar}{\h+0.2} 
		\draw (1) edge[out=-150,in=0]  (-90:\h) (-90:\h) edge[->,out=180,in=-30] (2);
		\draw (-90:\ar) node {$a$};
		\draw (2) edge[out=30,in=180]  (90:\h)  (90:\h) edge[->,out=0,in=150] (1);
		\draw (90:\ar) node {$a,b$};
		\else 
		\pgfmathsetmacro{\ar}{\r*cos(0.5*(360-\g)/\n)+0.2} 
		\pgfmathsetmacro{\abr}{\r*cos(0.5*((360-\g)/\n+\g))+0.2} 
		\foreach \i in {1,...,\nmm}{
			\pgfmathsetmacro{\angle}{90-0.5*\g-(360-\g)*(\i+0.5)/\n}
			\node[mynode] (ipp) at (\angle:\r) {};
			\draw[->] (\i) -- (ipp);
		}
		\foreach \i in {1,...,\nmcmm}{
			\pgfmathsetmacro{\angle}{90-0.5*\g-(360-\g)*\i/\n} 
			\draw (\angle:\ar) node {$a$};
		}
		\draw[->] (\n) -- (1);
		\draw (90:\abr) node {$a,b$};
		\fi
		\ifx0#3 
		\else 
		\foreach \i in {\nmc,...,\nmm}{
			\pgfmathsetmacro{\angle}{90-0.5*\g-(360-\g)*\i/\n} 
			\draw (\angle:\ar) node {$b$};
		}
		\fi
		\foreach \i in {1,...,\nmcmm} {
			\pgfmathsetmacro{\angle}{90-0.5*\g-(360-\g)*(\i-0.5)/\n}
			\pgfmathsetmacro{\outangle}{\angle+30}
			\pgfmathsetmacro{\inangle}{\angle-30}
			\draw (\i) edge[->,out=\outangle,in=\inangle,looseness=7] (\i);
		}
		\foreach \i in {1,...,\n} {
			\pgfmathsetmacro{\angle}{90-0.5*\g-(360-\g)*(\i-0.5)/\n}
			\pgfmathsetmacro{\opa}{(\i<\n-\c)?1:0}
			\draw[opacity=\opa] (\i) + (\angle:0.9) node {$b$};
		} 
	\end{scope}
}

\newcommand{\Cernystar}[4][0]{\begin{scope}[semithick]
		\pgfmathsetmacro{\g}{#1} 
		\pgfmathsetmacro{\n}{#2} 
		\pgfmathsetmacro{\c}{#3} 
		\pgfmathsetmacro{\r}{#4} 
		\pgfmathsetmacro{\nmm}{\n-1}
		\pgfmathsetmacro{\nmc}{\n-\c}
		\pgfmathsetmacro{\nmcmm}{\n-\c-1}
		\foreach \i in {1,...,\n}{
			\pgfmathsetmacro{\angle}{90-0.5*\g-(360-\g)*(\i-0.5)/\n}
			\pgfmathsetmacro{\CernyCol}{\i<\n-\c?"CernyColA":(\i<=\n?"CernyColB":"CernyColC")}
			\node[mynode,draw,fill=\CernyCol] (\i) at (\angle:\r) {{$\i$}};
		}
		\ifx2#2 
		\pgfmathsetmacro{\h}{(\r+0.1)/3} 
		\pgfmathsetmacro{\ar}{\h+0.22} 
		\draw (1) edge[out=-150,in=0]  (-90:\h) (-90:\h) edge[->,out=180,in=-30] (2);
		\draw (-90:\ar) node {$a,\tilde{a}$};
		\draw (2) edge[out=30,in=180]  (90:\h)  (90:\h) edge[->,out=0,in=150] (1);
		\draw (90:\ar) node {$\tilde{a},\tilde{b}$};
		\else 
		\pgfmathsetmacro{\ar}{\r*cos(0.5*(360-\g)/\n)} 
		\pgfmathsetmacro{\abr}{\r*cos(0.5*((360-\g)/\n+\g))+0.22} 
		\foreach \i in {1,...,\nmm}{
			\pgfmathsetmacro{\angle}{90-0.5*\g-(360-\g)*(\i+0.5)/\n}
			\node[mynode] (ipp) at (\angle:\r) {};
			\draw[->] (\i) -- (ipp);
		}
		\foreach \i in {1,...,\nmcmm}{
			\pgfmathsetmacro{\angle}{90-0.5*\g-(360-\g)*\i/\n} 
			\begin{scope}[shift={(\angle:\ar)},xscale=1.7]
				\draw (\angle:0.22) node {$a,\tilde{a}$};
			\end{scope}
		}
		\draw[->] (\n) -- (1);
		\draw (90:\abr) node {$\tilde{a},\tilde{b}$};
		\fi
		\ifx0#3 
		\else 
		\foreach \i in {\nmc,...,\nmm}{
			\pgfmathsetmacro{\angle}{90-0.5*\g-(360-\g)*\i/\n} 
			\draw (\angle:\ar) node {$b$};
		}
		\fi
		\foreach \i in {1,...,\nmcmm} {
			\pgfmathsetmacro{\angle}{90-0.5*\g-(360-\g)*(\i-0.5)/\n}
			\pgfmathsetmacro{\outangle}{\angle+30}
			\pgfmathsetmacro{\inangle}{\angle-30}
			\draw (\i) edge[->,out=\outangle,in=\inangle,looseness=7] (\i);
		}
		\foreach \i in {1,...,\n} {
			\pgfmathsetmacro{\angle}{90-0.5*\g-(360-\g)*(\i-0.5)/\n}
			\pgfmathsetmacro{\opa}{(\i<\n-\c)?1:0}
			\draw[opacity=\opa] (\i) + (\angle:0.9) node {$\tilde{b}$};
		} 
	\end{scope}
}

\begin{document}

\ifijfcs	
\markboth{S. Cambie, M. de Bondt, and H. Don}{Extremal Binary PFAs with Small Number of States}
%
\catchline{}{}{}{}{}
%
\title{Extremal Binary PFAs with Small Number of States\footnote{The first author has been supported by a Vidi Grant of the Netherlands Organization for Scientific Research (NWO), grant number $639.032.614$. His current affiliation is at the Extremal Combinatorics and Probability Group (ECOPRO), Institute for Basic Science (IBS), Daejeon, South Korea.}
}
\author{Stijn Cambie and Michiel de Bondt and Henk Don}
\address{Department of Mathematics, Radboud University Nijmegen, Postbus 9010, 6500 GL Nijmegen, The Netherlands\\
	\email{stijn.cambie@hotmail.com,\{m.debondt, h.don\}@math.ru.nl}}	
\else
\title{Extremal Binary PFAs with Small Number of States 
	\thanks{The first author has been supported by a Vidi Grant of the Netherlands Organization for Scientific Research (NWO), grant number $639.032.614$. His current affiliation is at the Extremal Combinatorics and Probability Group (ECOPRO), Institute for Basic Science (IBS), Daejeon, South Korea.}
}
\author{Stijn Cambie \and Michiel de Bondt \and Henk Don}
\authorrunning{S. Cambie et al.}
\institute{Department of Mathematics, Radboud University Nijmegen, Postbus 9010, 6500 GL Nijmegen, The Netherlands\\ \email{stijn.cambie@hotmail.com, \{m.debondt, h.don\}@math.ru.nl}}	
\fi

\maketitle              
	
	\begin{abstract}
		The largest known reset thresholds for DFAs are equal to $(n-1)^2$, where $n$ is the number of states. This is conjectured to be the maximum possible. PFAs (with partial transition function) can have exponentially large reset thresholds. This is still true if we restrict to binary PFAs. However, asymptotics do not give conclusions for fixed $n$. We prove that the maximal reset threshold for binary PFAs is strictly greater than $(n-1)^2$ if and only if $n\geq 6$. 
		
		These results are mostly based on the analysis of synchronizing word lengths for a certain family of binary PFAs. This family has the following properties: it contains the well-known \v{C}ern\'y automata; for $n\leq 10$ it contains a binary PFA with maximal possible reset threshold; for all $n\geq 6$ it contains a PFA with reset threshold larger than the maximum known for DFAs. 
		
		Analysis of this family reveals remarkable patterns involving the Fibonacci numbers and related sequences such as the Padovan sequence. We derive explicit formulas for the reset thresholds in terms of these recurrent sequences.
		
		Asymptotically the \v{C}ern\'y family gives reset thresholds of polynomial order. We prove that PFAs in the family are not extremal for $n\geq 41$. For that purpose, we present an improvement of Martyugin's prime number construction of binary PFAs. 
\end{abstract}


	\section{Introduction and Preliminaries}

\ifijfcs		
\begin{wrapfigure}{r}{4cm}  	
\begin{tikzpicture}
\useasboundingbox (-1,-0.5) rectangle (0,2);
\begin{scope}[shift={(1,1)}]
\Cerny{4}{0}{1.41421356237}
\end{scope}
\node at (1,-0.9) {The DFA $C_4$};
\end{tikzpicture}	
\end{wrapfigure}
\else
\begin{wrapfigure}{r}{3.5cm}  	
\begin{tikzpicture}
\useasboundingbox (-1,0) rectangle (0,0.1);
\begin{scope}[shift={(0.707,0.707)}]
\Cerny{4}{0}{1}
\end{scope}
\node at (0.707,-0.9) {The DFA $C_4$};
\fi
\end{tikzpicture}	
\end{wrapfigure}
	
	
	The diagram on the right depicts the \emph{deterministic finite automaton} (DFA) $C_4$. Starting in any state $q$ and reading the word $ba^3ba^3b$ leads to state 1. Therefore, $w$ is called a synchronizing word for $C_4$. It is also the only synchronizing word for $C_4$ of length at most 9. 
	
	Formally, a DFA $A$ is defined as a triple $(Q,\Sigma,\delta)$.  Here $\Sigma$ is a finite alphabet, $Q$ a finite set of states, which we generally choose to be $[n]=\{1,2,\ldots,n\}$, and $\delta: Q \times \Sigma \to Q$ the transition function. For $w \in \Sigma^*$ and $q \in Q$, we define $qw$ inductively by
	$q\varepsilon = q$ and $q w a = \delta(qw,a)$ for
	$a \in \Sigma$, where $\varepsilon$ is the empty word.
	So $qw$ is the state where one ends, when starting in $q$ and reading the symbols in $w$ consecutively, and $qa$ is a shorthand notation for
	$\delta(q,a)$. We extend the transition function to sets $S\subseteq Q$ by $Sw := \{qw:q\in S\}$. A word $w \in \Sigma^*$ is called {\em synchronizing}, if a
	state $q_s \in Q$ exists such that $q w = q_s$ for all $q \in Q$. The length of a shortest word with this property is the \emph{reset threshold} of $A$. 
	
	A central conjecture in the field is the following. It is attributed to \v{C}ern\'y's paper \cite{C64} of 1964, but a more accurate acknowledgement can be found in \cite{volkov3}.
	\begin{conjecture}\label{cernyconjecture}
		Every synchronizing DFA on $n$ states admits a synchronizing word of length $\leq (n-1)^2$.
	\end{conjecture}
	
	We denote the maximal possible reset threshold for a DFA on $n$ states by $d(n)$, rephrasing the conjecture to $d(n)=(n-1)^2$. The best known upper bounds are still cubic in $n$. In 1983 Pin \cite{pin} established the bound $\frac{1}{6}(n^3-n)$, using a combinatorial result by Frankl \cite{frankl}. 
	More than thirty years later, the leading constant was improved to 0.1664 by Szyku\l{}a, and subsequently to 0.1654 by Shitov \cite{shitov}.
	For a survey on synchronizing automata and the \v{C}ern\'y conjecture, we refer to \cite{volkov}.
	
	If Conjecture \ref{cernyconjecture} holds true, the bound is sharp. The DFA $C_4$ is one in a sequence found by \v{C}ern\'y \cite{C64}. For $n\geq 2$, the DFA $C_n$ has $n$ states which we denote by $Q=[n]$, a symbol $a$ sending $q$ to $q+1 \pmod n$ and a symbol $b$ sending $n$ to 1 and being the identity in all other states. The shortest synchronizing word for $C_n$ is $b(a^{n-1}b)^{n-2}$ of length $(n-1)^2$, so that $d(n)\geq (n-1)^2$.
	
	The picture changes drastically if we consider \emph{partial finite automata} (PFAs). In a PFA, the transition function is allowed to be partial. This means that $qa$ may be undefined for $q\in Q$ and $a\in\Sigma$. If $q\in S\subseteq Q$ and $qw$ is undefined, then $Sw$ is undefined as well. In this setting a word $w$ is called synchronizing for a PFA if there exists a $q_s\in Q$ such that $qw$ is defined and $qw=q_s$ for all $q\in Q$. 
Our notion of synchronization for PFAs is equivalent to D1- and D3-direction, and to careful synchronization as in \S 6.2 of \cite{V16}, but not to D2-direction and exact synchronization \cite{Shabana_2019}. The last two notions allow $qw$ to be undefined.


	For PFAs the maximal reset thresholds grow asymptotically like an exponential function of $n$, in contrast with the polynomial growth for DFAs. Also the behaviour in terms of alphabet size is different. The upper bound of Conjecture \ref{cernyconjecture} is attained by binary DFAs. For PFAs there is evidence that the alphabet size has to grow with $n$ to attain the maximal reset thresholds \cite{BDZ19}. Still, also binary PFAs give exponentially growing reset thresholds. We denote the maximal values by $p(n,2)$. A binary PFA attaining the maximal reset threshold is called \emph{extremal}. For $2\leq n \leq 10$ the values as found in \cite{BDZ19} are given below. For $n\geq 11$, the maximum is unknown.
\begin{center}	
\begin{tikzpicture}[x=6mm,y=-5mm]
\fill[tableIndex] (0,-1) rectangle (9,0);
\fill[tableFrame] (-2.45,-1) rectangle (0,1);
\draw[white] (-2.45,0) -- (0,0);
\node[white,anchor=east] at (0,-0.44) {{\mathversion{bold}$n\,\mathstrut$}};
\node[white,anchor=east] at (0,0.56) {{\mathversion{bold}$p(n,2)\mathstrut$}};
\foreach \n/\p in {2/1, 3/4, 4/9, 5/16, 6/26, 7/39, 8/55, 9/73, 10/94} {
  \draw[tableFrame,very thick] (\n-2,-1) -- (\n-2,1);
  \node at (\n-1.5,-0.44) {$\n\mathstrut$};
  \node at (\n-1.5,0.56) {$\p\mathstrut$};
}
\draw[tableFrame,very thick] (-2.45,-1) -- (9,-1) -- (9,1) -- (-2.45,1) -- cycle (0,0) -- (9,0) ; 
\end{tikzpicture}
\end{center}
	For all $2\leq n\leq 10$, these reset thresholds are attained by members of what we will call the \emph{\v{C}ern\'y family}. This family of PFAs $C_n^c$ will be introduced in Section~\ref{sec:cernyfamily}. 

In Section~\ref{sec:race} we relate the problem of finding reset thresholds for this family to a minimization problem involving racing pawns. A recursive solution for this problem is presented in Section~\ref{sec:recursive_asymptotic}, from which it follows that the maximal reset thresholds in the family grow like $n^2\log(n)$. In Section \ref{sec:solve_fc}, we give an exact solution of the minimization problem in terms of recurrent sequences. In addition, we determine the number of different optimal races. 
	In Section~\ref{sec:est_fc} we estimate the solution more precisely and find the asymptotic size of a shortest synchronizing word for $C_n^c$ for fixed $c$. Furthermore, we estimate the optimal choice of $c$ asymptotically in terms of $n$.
\ifijfcs
\else	
In Section~\ref{sec:drops}, we discuss odd behavior in the optimal choice of $c$ which emerged from computations. 
\fi 
	
	We end with the presentation of another construction of binary PFAs in Section~\ref{sec:Primes}, to defeat the \v{C}ern\'y family for large $n$. We show in Section~\ref{sec:LargerRT} that this construction gives binary PFAs with larger reset thresholds than $C_n^c$ for $n \ge 41$. It is unknown to us if there are constructions that beat the \v{C}ern\'y family for some $n<41.$ Our construction is an improvement of Martyugin's prime number construction of binary PFAs \cite{Mar08}, which has reset threshold $\exp{\big((1+o(1))\sqrt{n\ln(n)/2}\big)}$. We show in Section~\ref{sec:Primes} that the asymptotic behavior of the reset theshold of our construction is $\exp{\big((1+o(1))\sqrt{n\ln(n)}\big)}$, which is comparable to that of Martyugin's prime number construction of ternary PFAs \cite{Mar08}. We do this by providing sufficiently accurate estimates of the reset theshold for all three prime number constructions. To our knowledge, estimates with this level of accuracy have not been given before.
	
	The current paper extends the earlier work in \cite{BCD21}. The proof of Theorem \ref{thm:reduction} has been formalized and the content has been extended by the results in sections \ref{sec:solve_fc}, \ref{sec:est_fc}, \ifijfcs\else\ref{sec:drops},\fi \ref{sec:Primes} and \ref{sec:LargerRT}.

	\section{Extending the \v{C}ern\'y Sequence to a Family}\label{sec:cernyfamily}
	The \v{C}ern\'y family of binary PFAs, denoted by $C_n^c$, contains the \v{C}ern\'y sequence $C_n = C_n^0$ of binary DFAs. For fixed $c\in\mathbb{N}$ and $n\geq c+2$, we define the PFA $C_n^c$ with $n$ states and alphabet $\Sigma=\{a,b\}$ by
	\[
	qa = \left\{ \begin{array}{ll} 
	q+1\quad & 1 \le q \le n-c-1 \\
	\bot     & n-c \le q \le n-1 \\
	1        & q = n
	\end{array} \right. \qquad
	qb = \left\{ \begin{array}{ll} 
	q\quad & 1 \le q \le n-c-1 \\
	q+1  \quad      & n-c \le q \le n-1\\
	1        & q = n
	\end{array} \right.
	\]
	The PFA $C_n^c$ is depicted in Figure~\ref{fig:PFA_Pn^c} for $n=8$ and $c=2$, next to the DFA $C_n^{0}$ of \v{C}ern\'y.
	By analyzing this family, we obtain our main results.
	In particular, we will conclude that $ p(n,2)>(n-1)^2$ if and only if $n\geq 6$.
\ifijfcs
\else
\newpage
\fi	
	\begin{figure}[ht!]
		\centering{
			\begin{tikzpicture}
			\Cerny{8}{0}{2}
			\begin{scope}[shift={(6,0)}]
			\Cerny{8}{2}{2}
			\end{scope}
			\end{tikzpicture}
		}
		\caption{The DFA $C_8^0$ and the PFA $C_8^2$}
		\label{fig:PFA_Pn^c}
	\end{figure}
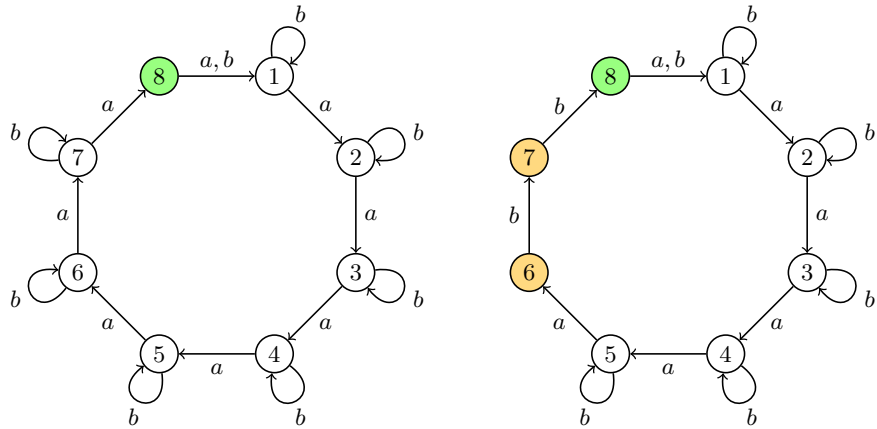

	Before deriving general formulas for the reset thresholds, we present the values for $2\leq n\leq 15$ and $0\leq c\leq 4$ in the following table. Independent of the analysis that will follow, these values were found by an algorithm computing the reset threshold for a given PFA.
\begin{center}	
\begin{tikzpicture}[x=6mm,y=-5mm]
\fill[tableIndex] (0,-1) rectangle (15.5,0);
\fill[tableFrame] (-2,-1) rectangle (0,5);
\draw[white] (-2,0) -- (0,0);
\node[white,anchor=east] at (0,-0.44) {{\mathversion{bold}$n\,\mathstrut$}};
\foreach \c [count=\cpp] in { 0, 1, 2, 3, 4 } {
  \node[white,anchor=east] at (0,\cpp-0.44) {{\mathversion{bold}$c=\c\mathstrut$}};
  \fill[tableVoid] (0,\c) rectangle (\c,\cpp);
}
\fill[tableBold] (0,0) rectangle (4,1) rectangle (8,2) rectangle (9,3);
\foreach \n/\subopts/\opts in {2/{{}}/{1}, 3/{{},2}/{4}, 4/{{},7,3}/{9}, 5/{{},15,10,4}/{16}, 6/{25,{},21,13,5}/{{},26}, 7/{36,{},35,27,16}/{{},39}, 8/{49,{},52,44,33}/{{},55}, 9/{64,{},72,65,53}/{{},73}, 10/{81,93,{},89,78}/{{},{},94}} {
  \draw[tableFrame,very thick] (\n-2,-1) -- (\n-2,5);
  \node at (\n-1.5,-0.44) {$\n\mathstrut$}; 
  \foreach \subopt [count=\cpp] in \subopts {
    \node at (\n-1.5,\cpp-0.44) {$\subopt\mathstrut$};
  }
  \foreach \opt [count=\cpp] in \opts {
    \node at (\n-1.5,\cpp-0.44) {{\mathversion{bold}$\opt\mathstrut$}};
  }
}
\begin{scope}[shift={(9,0)},xscale=1.3,shift={(-9,0)}]
\fill[tableBold] (9,2) rectangle (12,3) (11,3) rectangle (14,4);
\foreach \n/\sopts/\opts in {11/{100,116,{},115,106}/{{},{},119}, 12/{121,141,{},144,136}/{{},{},146}, 13/{144,168,{},{},169}/{{},{},176,176}, 14/{169,197,208,{},206}/{{},{},{},211}, 15/{196,228,242,{},246}/{{},{},{},248}} {
  \draw[tableFrame,very thick] (\n-2,-1) -- (\n-2,5);
  \node at (\n-1.5,-0.44) {$\n\mathstrut$};  
  \foreach \sopt [count=\cpp] in \sopts {
    \node at (\n-1.5,\cpp-0.44) {$\sopt\mathstrut$};
  }
  \foreach \opt [count=\cpp] in \opts {
    \node at (\n-1.5,\cpp-0.44) {{\mathversion{bold}$\opt\mathstrut$}};
  }
}
\end{scope}
\draw[tableFrame,very thick] (-2,-1) -- (15.5,-1) -- (15.5,5) -- (-2,5) -- cycle (0,0) -- (15.5,0);
\draw[tableFrame] foreach \c in {1,2,3,4} { (0,\c) -- (15.5,\c) };
\end{tikzpicture}
\end{center}
	Values in boldface represent the maximal reset threshold in the family for a given $n$. For $n=13$, the maximum is attained twice, see also Figure~\ref{fig:C13_2&3}. Later we will see that for large $n$, the optimal $c$ is close to $n/2$. For $2\leq n\leq 10$, these maxima exactly match the values of $p(n,2)$. This means that the \v{C}ern\'y family contains a binary PFA on $n$ states with maximal possible reset threshold for all $2\leq n\leq 10$. In fact, for $6\leq n\leq 10$, there exists only one binary PFA reaching this maximum \cite{BDZ19}.
	
	The first line of the table shows the squares $(n-1)^2$ for the \v{C}ern\'y sequence $C_n^0$. To give explicit expressions for subsequent lines is much harder. The order of growth is still quadratic for every $c$, but no formula of the form $a_2n^2+a_1n+a_0$ exists in general, as we will see later in this paper. 
	
\begin{figure}\label{fig:C13_2&3}
\centering{
\begin{tikzpicture}
\Cerny[9]{13}{2}{2}
\begin{scope}[shift={(6,0)}]
\Cerny[9]{13}{3}{2}
\end{scope}
\end{tikzpicture}
}
\caption{The PFAs $C^2_{13}$ and $C^3_{13}$ both synchronize in $176$ steps.}
\end{figure}
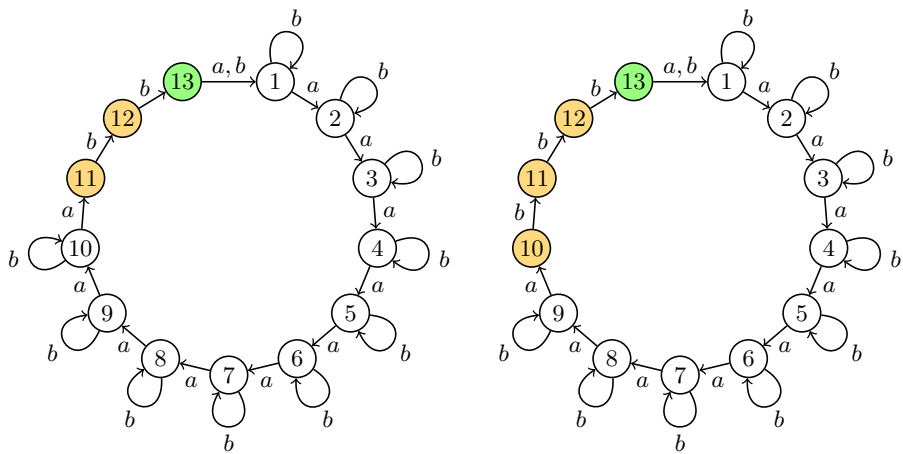

	We now turn to the analytic derivation of reset thresholds for the \v{C}ern\'y family. We use the following interpretation of synchronization: let a pawn be placed in every state of a PFA, let them simultaneously follow the same word $w$ and let two of them merge if they are in the same state after reading some prefix of $w$. A synchronizing word is then a word that merges all pawns.

	\section{Reduction to a Pawn Race Problem}\label{sec:race} 

Our first result reduces the question of synchronizing $C_n^c$ to the following problem.
\begin{problem}[Pawn race problem]\label{prob:pawnproblem}
	We have $n$ pawns on the \emph{integers} $1,2,\ldots ,n$.
	In every iteration, every pawn has the choice to move from its location $k$ to $k+1$ or to stay at $k$. Moving costs $c+1$, staying costs $c$. After every iteration, if two pawns are in the same position, they merge. 
	What is the minimum cost for which it is possible to merge all the pawns?	
\end{problem}

\begin{theorem}\label{thm:reduction} Let $f_c(n)$ be the solution to Problem \ref{prob:pawnproblem} and denote $n' = n-c-1$. The reset threshold of $C_n^c$ is equal to
	\[
	n'(n'-1)+c+1+f_c(n').
	\]
\end{theorem}

The rest of this section will be devoted to the proof of Theorem \ref{thm:reduction}.

\begin{lemma} \label{clm:sync}
	$C_n^c = (Q,\Sigma,\delta_n)$ has a synchronizing word. 
\end{lemma}

\begin{proof}
	We denote $[k] := \{1,2,\ldots,k\}$ and note that $Qb^{c+1} = [n-c-1] \subset [n-c]$. Define $\tilde a = b^ca$ and $\tilde b = b^{c+1}$. Then $\tilde a$ acts as a cyclic permutation on $[n-c]$ and $\tilde b$ sends $n-c$ to 1 and is the identity otherwise. Here we recognize the \v{C}ern\'y automaton $C_{n-c}^0$, 
	so that $C_n^c$ is synchronizing.
	\qed\end{proof}

Inspired by the proof of Lemma \ref{clm:sync}, we define the PFA $C_{n}^{*}=([n],\Gamma,\eta_n)$ with state set $[n]$ and alphabet $\Gamma = \{a, \tilde{a},\tilde{b}\}$. The transition function is defined by 
\begin{align}
\begin{cases}qa=q\tilde a = q+1\ \text{and}\ q\tilde b = q &\qquad \text{if}\ q\neq n,\\
na=\bot, n\tilde a =n\tilde b= 1.
\end{cases}
\end{align}
See Figure \ref{fig:PFA_C*} for an illustration. Observe that restricting the transition function $\delta_n$ of $C_n^c$ to $[n-c]$ relates to the PFA $C_{n-c}^*$ in the following way:
\begin{align}\label{eq:transitionrelation}
\delta_n(q,a)=\eta_{n-c}(q,a),\qquad \delta_n(q,b^ca)=\eta_{n-c}(q,\tilde a),\qquad \delta_n(q,b^{c+1})=\eta_{n-c}(q,\tilde b) 
\end{align}
for all $q\in [n-c]$. By substituting $\tilde a = b^ca$ and $\tilde b = b^{c+1}$, a word $w\in\Gamma^*$ naturally corresponds to a word $s_c(w)\in\Sigma^*$ with the property $\delta_n(q,s_c(w)) = \eta_{n-c}(q,w)$ for all $q\in [n-c]$. We define the {\em $c$-weighted length}
of a word $w\in\Gamma^*$ as the length of $s_c(w)$, which we denote by $|w|_c$. 

In Corollary \ref{cor:1.4} below, we prove that a synchronizing word of minimal $c$-weighted length for $C_{n-c}^*$ corresponds to a shortest synchronizing word for $C_{n}^c$. For instance, consider the PFA $C_6^*$ as in Figure \ref{fig:PFA_C*} and take $c=0$. Then $\tilde a=a$ and $\tilde b=b$ have weight $1$ and the resulting PFA is equivalent to the \v{C}ern\'y automaton $C_6^0$. If we instead take $c=2$, then $\tilde a=b^2a$ and $\tilde b=b^3$ both have weight 3 and a word of minimal $2$-weighted length for $C_6^*$ corresponds to a shortest synchronizing word for the PFA $C_8^2$ given in Figure \ref{fig:PFA_Pn^c}.

\begin{figure}[ht!]
	\vspace{-7pt}
	\centering{
		\begin{tikzpicture}
		\Cerny{6}{0}{1.7}
		\begin{scope}[shift={(6,0)}]
		\Cernystar{6}{0}{1.7}
		\end{scope}
		\end{tikzpicture}
	}
	\caption{The DFA $C_6^0$ and the PFA $C_6^*$. If $\tilde a$ and $\tilde b$ have weight 3, and $a$ has weight 1, then a synchronizing word of minimum weighted length for $C_6^*$ corresponds to a shortest synchronizing word for $C_8^2$.}
	\label{fig:PFA_C*}
	\vspace{-7pt}		
\end{figure}
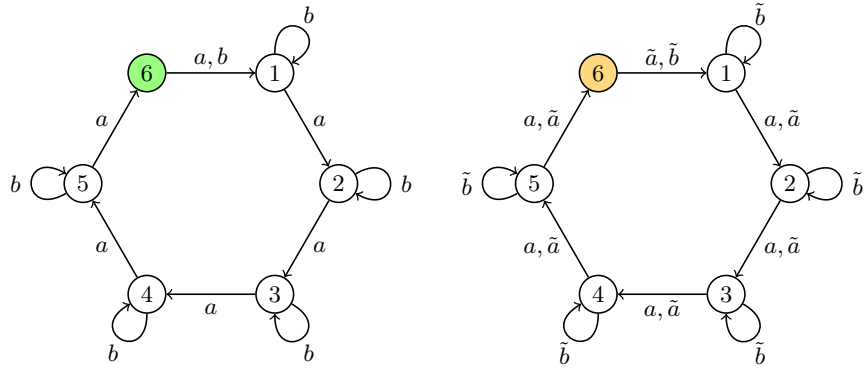

Let $S \subseteq Q$ and $w \in \Sigma^{*}$. We say that \emph{$w$ has minimum length for $Sw = T$} if $Sw = T$, and if $Sv = T$ implies $|v| \ge |w|$ for all $v \in \Sigma^{*}$. 

\begin{lemma}\label{clm:obs1}
	Consider $C_n^c$. Let $S \subseteq Q=[n]$ and $T \subsetneq [n-c]$. Take $w \in \Sigma^{*}$ such that $w \ne \varepsilon$, and suppose that 
	$w$ has minimum length for $Sw = T$. 
	\begin{enumerate}[(i)]
		\item\label{claim1.1} If $S=Q$, then $w$ starts with $b^{c+1}$.
		\item\label{claim1.2} If $S \subseteq [n-c]$ and $n-c\in S$, then $w$ starts with $b^ca$ or $b^{c+1}$.
		\item\label{claim1.3} If $S \subseteq [n-c]$ and $n-c\not\in S$, then $w$ starts with $a$.
	\end{enumerate}	
\end{lemma}

\begin{proof} 
	Since $qb^{c+2} = qb^{c+1}$ for all states $q \in Q$, the word $w$ cannot contain $b^{c+2}$. Furthermore, $(n-c)b^{m}a$ is not defined for $m=0,1,\ldots, c-1$. 
	Using the fact that $(n-c) b^m \notin [n-c]$ for $m=1,2,\ldots,c$, 
	statement (ii) follows.
		
	It also follows that $w$ starts with $b^c$ if $S=Q$. Suppose that $S=Q$ and $w$ starts with $b^ca$. Since $Qb^ca = [n-c]$, we infer from (ii) that $w$ starts with $b^cab^c$. But this contradicts the assumption that $w$ has minimum length, because $Qb^{c} = [n-c-1] \cup \{n\} = Qb^cab^c$. This yields (i).
	
	Statement (iii) follows by observing that $Sb=S$ in this case. 
	\qed\end{proof}

This lemma gives a corollary to relate words in $C_n^c = ([n],\{a,b\},\delta_n)$ to words in $C_{n-c}^*=([n-c],\{a,\tilde a,\tilde b\},\eta_{n-c})$. Essentially, for $S\subseteq [n-c]$ there is a one-to-one correspondence between a word $w$ of minimum length for $\delta_n(S,w)=T$ and a word $w'$ of minimum $c$-weighted length for $\eta_{n-c}(S,w')=T$.

\begin{corollary}\label{cor:1.4}
	Let $c \in \N$, and suppose that $S = [n]$ or $S \subseteq [n-c]$. Let $T \subsetneq [n-c]$. Suppose that $w \in \Sigma^{*}$ has minimum length for $\delta_n(S,w) = T$. Then there exists $w'\in\Gamma^{*}$ with $s_c(w') = w$ and minimum $c$-weighted length for $\eta_{n-c}([n-c]\cap S,w')=T$.
\end{corollary}

\begin{proof}
	If $w=\varepsilon$, then $S = T \subsetneq [n-c]$. Therefore, $[n-c]\cap S=T$ and $w'=\varepsilon$ suffices. So assume that $w \ne \varepsilon$. From Lemma \ref{clm:obs1}, it follows that $w$ has a prefix $u \in \{a,b^c a,b^{c+1}\}$. Let $u'\in\Gamma$ satisfy $s_c(u')=u$. By (\ref{eq:transitionrelation}), we have $\delta_n(S,u) = \eta_{n-c}([n-c]\cap S,u')\subseteq [n-c]$ if $S\subseteq [n-c]$. The same is true if $S=Q$, since then $u=b^{c+1}$. By induction, we find $w'$ such that $s_c(w') = w$ and $\eta_{n-c}([n-c]\cap S,w')=T$. 
	
	If $v'\in\Gamma^*$ were a word with $\eta_{n-c}([n-c]\cap S,v')=T$ and $|v'|_c<|w'|_c$, then we would have $\delta_n(S,s_c(v')) =T$ and $|s_c(v')|<|w|$, contradicting the minimality of $|w|$. So $|w'|_c$ is minimal.

\end{proof}

%
%

We will now consider $C_n^*$ and investigate words which are applied on a subset of the state set $[n]$. We see a subset of the state set $[n]$ as a collection of pawns on those states. Symbols $a$ and $\tilde{a}$ move these pawns clockwise without merging, but if both $n$ and $1$ are occupied by a pawn, then symbol $\tilde{b}$ merges both pawns. This is the only possibility for pawns to merge. We call a pawn a \emph{chaser} if its next merge will be with a pawn in front of it, and a resigner otherwise. So a chaser is on state $n$ directly before merging and a resigner on state $1$. 

Notice that pawns do not need to be a chaser or a resigner, because they may not take part in a merge. But if the word at hand is a synchronizing word, then all pawns end in the same state, which is state $1$ if the word has minimum $c$-weighted length.

Suppose for now that the word at hand is synchronizing. Since all pawns end in the same state, there is a unique pawn travelling the largest distance. This pawn is always a chaser. We therefore call it the \emph{lanterne rouge}. Similarly, there is a pawn that makes the least number of moves and is always a resigner. This one is called the \emph{yellow jersey}. If the lanterne rouge or the yellow jersey merges, then we see the pawn which results from the merge as its continuation. Therefore, the lanterne rouge and the yellow jersey are the chaser and resigner respectively of the last merge. 

\begin{lemma}\label{claim:n-c}
	Let $c \in \N$ and $S,T \subseteq [n]$, such that $n \in S$. Suppose that $w \in \Gamma^{*}$ has minimum $c$-weighted length for $Sw=T$.
	\begin{enumerate}[(i)]
		
		\item If the pawn at $n$ is a resigner, then $w$ starts with $\tilde{a}$.
		
		\item If the pawn at $n$ is a chaser and $c \ne 0$, then $w$ starts with $\tilde{b}$.
		
		\item If the pawn at $n$ is a chaser and $c=0$, then $w$ can be chosen to start with $\tilde{b}$.
		
	\end{enumerate} 
\end{lemma}

The intuition behind Lemma \ref{claim:n-c} is that it is optimal for both chasers and resigners to merge as quickly as possible. Loosely speaking, a chaser in state $n$ can get one step closer to its target by choosing $\tilde b$, while choosing $\tilde a$ would mean that the other pawns move as well so that the chaser makes no progress. Choosing $\tilde b$ therefore minimizes the time to merge and the $c$-weighted word length. 

On the other hand, one can say that a resigner in state $n$ gets one step farther from its chaser by choosing $\tilde b$, while choosing $\tilde a$ would mean that the other pawns move as well so that the resigner does not move relatively to the other pawns. Choosing $\tilde a$ therefore minimizes the time to merge and the $c$-weighted word length. 

We start with an informal setup of the proof of Lemma \ref{claim:n-c}, which we elaborate in Lemmas \ref{lem:resigner} and \ref{lem:chaser} below. If $w$ starts with $\tilde{b}$, then by Lemma \ref{clm:obs1}(\ref{claim1.3}), either $w = \tilde{b}$ or $w$ starts with $\tilde{b}a$. The effect of $\tilde{a}$ and $\tilde{b}a$ is similar except for state $n$: $i\tilde{a} = i+1 = i\tilde{b}a$ if $i \ne n$, and $n\tilde{a} = 1 \ne 2 = n\tilde{b}a$.
So $\tilde{b}a$ places the pawn at $n$ on the successor of the state where it would be placed with $\tilde{a}$, costing $(c+2) - (c+1) = 1$ extra $c$-weighted word length. On the other hand, $\tilde{a}$ places the pawn at $n$ on the predecessor of the state where it would be placed with $\tilde{b}a$, saving $(c+2) - (c+1) = 1$ word length.

The idea is that the relative displacement which is initiated by $\tilde{b}a$ instead of $\tilde{a}$, or vice versa, can be preserved by adapting the word without adapting its $c$-weighted length, namely by adapting the order of symbols to match the new positioning of the displaced pawn. 

For a chaser, at the last moment when it is on state $n$ before merging, $\tilde{b}$ is applied. With a displacement to its successor, this $\tilde{b}$ can be skipped, so $c+1$ $c$-weighted word length can be saved. With the cost of the displacement, this adds up to saving $c$ $c$-weighted word length. So the chaser must choose $\tilde b$ for optimality if $c \ne 0$, and can choose $\tilde b$ without harm if $c = 0$. 

For a resigner, at the last moment when the pawn is on state $n$ before merging, $\tilde{a}$ is applied. With a displacement to its predecessor, this $\tilde{a}$ can be replaced by $a$, so $c$ $c$-weighted word length can be saved. With the saving of the displacement, this adds up to saving $c+1$ $c$-weighted word length. So the resigner must choose $\tilde{a}$ for optimality. 

Lemma \ref{lem:resigner} below shows that the required $c$-weighted word length drops by at least $c$ if a resigner is displaced to its predecessor or if it has disappeared by merging. The lemma also shows how to adapt the word in order to preserve displacement.

\begin{lemma} \label{lem:resigner}
	Let $S \subseteq [n]$, and assume that $w \in \Gamma^{*}$ is defined on $S$. Suppose that $i \in S$ contains a resigner for $w$. If $i \ne 1$, and
	$$
	S' = S \setminus \{i\}\qquad\text{or}\qquad
	S' = (S \setminus \{i\}) \cup \{i-1\},
	$$
	then there exists a word $w' \in \Gamma^{*}$ such that $S'w' = Sw$ and $|w'|_c \le |w|_c - c$.
\end{lemma}

\begin{proof}
	Assume the lemma holds for $|w|<k$. Take $|w|=k$ and let $w_j \in \Gamma$ be the $j$\textsuperscript{th} symbol of $w$ for all $j$. Since $i$ contains a resigner for $w$, implicit assumptions on the length of $w$ in the arguments below are justified. We distinguish $2$ cases:
	
	\begin{itemize}
		
		\item \emph{Case 1: $i \ne 1$ and $S' = S \setminus \{i\}$.} 
		
		Suppose first that $i \ne n$. Then $iw_1\neq 1$ and for all $w_1 \in \Gamma$,
		$$
		S'w_1 = Sw_1 \setminus \{iw_1\}.
		$$
		Taking $w_1'=w_1$, the result follows by induction on $|w|$.
		
		Suppose next that $i = n$. Then either $w_1 = \tilde{a}$ or $w_1 = \tilde{b}$. If $w_1 = \tilde{a}$ then make $w'$ from $w$ by replacing $w_1$ by $a$. If $w_1 = \tilde{b}$, then make $w'$ from $w$ by removing $w_1$. In both cases, $|w'|_c \le |w|_c - c$ and
		$$
		S'w' = S'w\qquad \mbox{and} \qquad S'w \cup \{iw\} = Sw.
		$$
		Since the pawn at $i$ is merged by $w$ (it is a resigner), it follows that $S'w' = S'w = Sw$, which gives the result.
		
		\item \emph{Case 2: $i \ne 1$, $S' = (S \setminus \{i\}) \cup \{i-1\}$, and $S' \ne S \setminus \{i\}$.}
		
		Then $i-1 \notin S$. Suppose first that $i \ne n$. Then $iw_1\neq 1$ and for all $w_1 \in \Gamma$, $(i-1)w_1 = iw_1 - 1$ and
		$$
		S'w_1 =  (Sw_1 \setminus \{iw_1\}) \cup \{iw_1 - 1\}.
		$$
		Taking $w_1'=w_1$, the result follows by induction on $|w|$.
		
		Suppose next that $i = n$. Then either $w_1 = \tilde{a}$ or $w_1 = \tilde{b}$. From $n - 1 = i - 1 \notin S$, we infer that $n \notin Sw_1$ and $Sw_1\tilde{b} = Sw_1$. So we may assume that $w_2 \ne \tilde{b}$.  If $w_1 = \tilde{a}$, then $i w_1 w_2 = 2$ and
		$$
		S'w_2w_1 =  (Sw_1w_2 \setminus \{2\}) \cup \{1\}
		$$
		If $w_1 = \tilde{b}$, then $1 \notin S$ because the pawn at $i$ is a resigner, so $2 \notin S'w_2w_1$. Therefore, we obtain the same assertions as in the case $w_1 = \tilde{a}$. Taking $w_1'w_2'=w_2w_1$, the result follows by induction on $|w|$.\qed
		
	\end{itemize}
\end{proof}

If state $1$ contains a resigner, then removing it will not decrease the required $c$-weighted word length if the resigner is about to merge with its chaser, to advance as a chaser. So the condition that $i \ne 1$ in Lemma \ref{lem:resigner} is necessary. 

Lemma \ref{lem:chaser} shows that the required $c$-weighted word length drops by at least $c+1$ if a chaser is displaced to its successor or if it has disappeared by merging. The lemma also shows how to adapt the word in order to preserve displacement.


\begin{lemma} \label{lem:chaser}
	Let $S \subseteq [n]$, and assume that $w \in \Gamma^{*}$ is defined on $S$. Suppose that $i \in S$ contains a chaser of $w$. If
	\begin{align*}
	S' = S \setminus \{i\}\qquad\text{or}\qquad i\ne n\ \text{and}\ S' = (S \setminus \{i\}) \cup \{i+1\},
	\end{align*}
	then there exists a word $w' \in \Gamma^{*}$ such that $S'w' = Sw$ and $|w'|_c = |w|_c - c - 1$.
\end{lemma}

\begin{proof}
	Assume the lemma holds for $|w|< k$. Take $|w| = k$, and let $w_j \in \Gamma$ be the $j$\textsuperscript{th} symbol of $w$ for all $j$. Since $i$ contains a chaser for $w$, implicit assumptions on the length of $w$ in the arguments below are justified. We distinguish $2$ cases:
	\begin{itemize}
		
		\item \emph{Case 1: $S' = S \setminus \{i\}$}.
		
		Suppose first that $i \ne n$. If $i = 1$ and $w_1 = \tilde{b}$, then $n \notin S$, because the pawn at $i$ is a chaser. In all cases, for all $w_1 \in \Gamma$,
		$$
		S'w_1 = Sw_1 \setminus \{iw_1\}.
		$$
		Taking $w_1'=w_1$, the result follows by induction on $|w|$.
		
		Suppose next that $i = n$. Then either $w_1 = \tilde{a}$, or $w_1 = \tilde{b}$. If $w_1 = \tilde{a}$, then $i w_1 = 1$ and
		$$
		S'w_1 = Sw_1 \setminus \{1\},
		$$
		and the result follows by induction on $|w|$ by taking $w_1'=w_1$. So assume that $w_1 = \tilde{b}$. Make $w'$ from $w$ by removing $w_1$. Then $|w'|_c = |w|_c-c-1$ and
		$$
		S'w' = S'w \qquad \mbox{and} \qquad S'w \cup \{iw\} = Sw.
		$$
		Since the pawn at $i$ is merged by $w$ (it is a chaser), it follows that $S'w' = S'w = Sw$, which gives the result. Note that $w'=\varepsilon$ is possible if $1\in S$.
		
		\item \emph{Case 2: $i \ne n$, $S' = (S \setminus \{i\}) \cup \{i+1\}$, and $S' \ne S \setminus \{i\}$.}
		
		Then $i+1 \notin S$. Suppose first that $i \ne n-1$. If $i = 1$ and $w_1 = \tilde{b}$, then $n \notin S$, because the pawn at $i$ is a chaser. In all cases, for all $w_1 \in \Gamma$, $(i+1)w_1 = iw_1 + 1$ and
		$$
		S'w_1 =  (Sw_1 \setminus \{iw_1\}) \cup \{iw_1 + 1\}.
		$$
		Let $w_1'=w_1$ and note that $iw_1\neq n$. The result follows by induction on $|w|$.
		
		Suppose next that $i = n-1$. As $n = i + 1 \notin S$, $S\tilde{b} = S$. So we may assume that $w_1 \ne \tilde{b}$. Furthermore, either $w_2 = \tilde{a}$ or $w_2 = \tilde{b}$. In all cases, $i w_1 w_2 = 1$ and
		$$
		S'w_2w_1 =  (Sw_1w_2 \setminus \{1\}) \cup \{2\}.
		$$
		Taking $w_1'w_2' = w_2w_1$, the result follows by induction on $|w|$.\qed
		
	\end{itemize}
\end{proof}

We are now ready to give a formal proof of  Lemma \ref{claim:n-c}.

\ifijfcs
	\begin{proof}[Proof of Lemma \ref{claim:n-c}]
\else
	\begin{proof}[of Lemma \ref{claim:n-c}]
\fi
	Let $w_j \in \Gamma$ be the $j$\textsuperscript{th} symbol of $w$ for all $j$. Suppose first that the pawn in state $n$ is a resigner, and that $w_1 = \tilde{b}$. Then $1\not\in S$ and $w_1w_2 = \tilde{b}a$. Let $S' = S \tilde{a} = (S\tilde ba\setminus \{2\})\cup\{1\}$. From Lemma \ref{lem:resigner} with $i = 2$, it follows that $S'w' = T$ for a word $w' \in \Gamma^{*}$ of $c$-weighted length at most $|w|_c - |w_1w_2|_c - c = |w|_c - 2c - 2$. So $Sv = T$ for the word $v=\tilde a w'$ 
	of $c$-weighted length at most $|w|_c - c - 1$. Contradiction.
	
	Suppose next that the pawn in state $n$ is a chaser, and that $w_1 = \tilde{a}$. Suppose additionally that $w_1=\tilde{a}$ is inevitable if $c = 0$. Let $S' = S \tilde{b}a = (S\tilde a\setminus\{1\})\cup\{2\}$.
	From Lemma \ref{lem:chaser} with $i = 1$, it follows that $S'w' = T$ for a word $w' \in \Gamma^{*}$ of $c$-weighted length at most $|w|_c - |w_1|_c - c - 1 = |w|_c - 2c - 2$. So $Sv = T$ for the word $v = \tilde baw'$ 
	of  $c$-weighted length at most $|w|_c - c$, which starts with $\tilde{b}$. Contradiction.
	\qed\end{proof}

After this excursion to the auxiliary automaton $C_n^*$, we return to the automaton $C_n^c$. Let a pawn start in each of these states. Also in this automaton pawns only move clockwise by steps of size 1. We define chasers, resigners, the lanterne rouge and the yellow jersey, analogous to the definitions for $C_n^*$. A direct implication of Corollary \ref{cor:1.4} and Lemma \ref{claim:n-c} is the following.

\begin{corollary}\label{cor:strategy} 
	Let $c \in \N$, and suppose that either $n \in S = Q$ or $n \in S \subseteq [n-c]$. Let $T \subsetneq [n-c]$ and $w \in \Sigma^{*}$, and suppose that 
	$w$ has minimum length for $Sw = T$. 
	\begin{enumerate}[(i)]
		
		\item If the pawn at $n$ is a resigner, then $w$ starts with $b^c a$.
		
		\item If the pawn at $n$ is a chaser and $c \ne 0$, then $w$ starts with $b^{c+1}$.
		
		\item If the pawn at $n$ is a chaser and $c=0$, then $w$ can be chosen to start with $b^{c+1}$.
		
	\end{enumerate} 
\end{corollary} 

On account of Lemma \ref{clm:obs1}(\ref{claim1.1}), a shortest synchronizing word is of the form $b^{c+1}w$. Since $Qb^{c+1} = [n-c-1] = [n']$, the start state subset of $w$ is $[n']$.

Suppose the yellow jersey starts in $j \in [n'].$
To simplify the further investigation, we will first look to the \emph{shortest full synchronizing word}, which is the shortest word that synchronizes all pawns into state $j$. 
Now consider an arbitrary set $S\subseteq [n']$. By Lemma \ref{clm:obs1}(\ref{claim1.3}) and Corollary \ref{cor:1.4}, a shortest full synchronizing word for $S$ starts with $a$ and can be partitioned into factors $a$, $b^ca$ and $b^{c+1}$. It does not contain $b^{c+2}$ and therefore has a prefix of the form
\begin{align} \label{iterationword}
w = aw_{n'}aw_{n'-1}a\ldots aw_3aw_2aw_1aw_n,\qquad \text{with}\qquad w_k\in\{\varepsilon,b^c,b^{c+1}\}.
\end{align}
A word of this form will be called an \emph{iteration word}. If $w$ is a prefix of a shortest full synchronizing word for $S$ and $wb$ is not, then we call $w$ an \emph{optimal iteration word}. 

\begin{lemma}\label{claim:iterationword} Let $w$ be an optimal iteration word for $S\subseteq [n']$, such that every suffix of $w$ follows the choice in Corollary \ref{cor:strategy}(iii) (if $c = 0$). Then 
	\begin{enumerate}[(i)]
		\item\label{claim3.1}
		For $k\in [n']\setminus S$, we have $w_k=\varepsilon$. For $k\in S$, we have $w_k=b^c$ or $w_k = b^{c+1}$. 
		\item\label{claim3.2} 
		For all $k\in S$,\[kw=\left\{
		\begin{array}{ll}
		k & \text{if}\ w_k =b^c,\\
		k+1 \mod n' \qquad& \text{if}\ w_k =b^{c+1}.
		\end{array}
		\right.
		\]			
		\item\label{claim3.3} If $w_{n'} = b^{c+1}$, then $w_n=b^{c+1}$. If $w_{n'}\neq b^{c+1}$, then $w_n=\varepsilon$.
	\end{enumerate}
\end{lemma}

	\begin{proof} The word $w$ has the following properties for $1\leq k \leq n'$:
		\begin{align}\label{eq:iterationword}
		kw = \left\{
		\begin{array}{lll}
		\bot & \text{if} & w_k = \varepsilon\ \text{and}\ c\neq 0\\
		k & \text{if} & w_k = b^c\\
		k+1 \qquad& \text{if} & w_k =b^{c+1}, k\neq n'.
		\end{array}
		\right.
		\end{align}
		For $k\in [n']$, $w_k$ can only affect a pawn in state $k$, so that $w_k=\varepsilon$ if $k\in [n']\setminus S$. For $k\in S$, since $kw$ has to be defined, it follows that $w_k\in\{b^c,b^{c+1}\}$ proving (i). Statement (ii) follows as well, except for the case where $k = n'$ and $w_{n'} = b^{c+1}$.
		
		To complete the proof of (ii), and to prove the first claim of (iii), suppose that $w_{n'} = b^{c+1}$. Then it follows from Corollary \ref{cor:strategy} that the pawn in state $n'$ is a chaser. Write $w=vw_n$. Then $n'v = n-c$ and $kv=kw\neq n-c$ for all $k\neq n'$. Therefore, the pawn under consideration did not merge yet and is still chasing. Using Corollary \ref{cor:strategy} again, combined with the assumption of following the choice in Corollary \ref{cor:strategy}(iii), we infer that $w_n=b^{c+1}$. This yields the first part of (\ref{claim3.3}). Therefore, $n'w = n'vw_n = (n-c)b^{c+1} = 1$. This completes the proof of (ii).
		
		If $w_{n'} = b^c$ or $w_{n'}=\varepsilon$, then $Sv\subseteq [n']$. By Lemma \ref{clm:obs1}(\ref{claim1.3}), a shortest synchronizing word for $S$ then starts with $va$ so that $w_n=\varepsilon$, completing the proof.
		\qed\end{proof}
	
\ifijfcs	
	\begin{proof}[Proof of Theorem \ref{thm:reduction}]
\else
	\begin{proof}[of Theorem \ref{thm:reduction}]
\fi
		The idea of Lemma \ref{claim:iterationword} is that an iteration word can be used to decide for every pawn if it has to move one step (at the cost of $c+1$ letters $b$ and possible more if we needed $w_n=b^{c+1}$), or to stay where it is (at the cost of $c$ letters $b$).
		The optimal choice depends on the pawn being a chaser or a resigner, where we follow the choice in Corollary \ref{cor:strategy}(iii) if $c = 0$.
		After applying an optimal iteration word, all pawns will be located on a subset of $[n']$. Consequently, every shortest full synchronizing word can be partitioned into iteration words. 
		
		As the yellow jersey starts in $j \in[n']$, the lanterne rouge starts in $j+1 \mod n'$. 
		Observe that after each iteration, the lanterne rouge (being a chaser) will have moved from $\ell$ to $\ell+1 \pmod {n'}$, while the yellow jersey is still at $j$. After $n'-1$ iterations, both the lanterne rouge and the yellow jersey and hence all initial pawns are in $j$.
		For the shortest synchronizing word, it is sufficient to have all pawns in state 1, so we can delete $a^{j-1}$ at the end of the shortest full synchronizing word. 
		
		Hence the number of letters $a$ in a shortest synchronizing word equals $(n'-1)\bcdot(n'+1)-(j-1)$. We have used $c+1$ letters $b$ in the beginning and at least $f_c(n')$ letters $b$ in all iteration words. By Lemma \ref{claim:iterationword}(\ref{claim3.3}), there is an additional cost of $c+1$ letters $b$ for each iteration word with $w_{n'}=b^{c+1}$. Now suppose the yellow jersey starts in $j=n'$. This minimizes the number of $a$'s. Furthermore, since the yellow jersey is always a resigner, $w_{n'}=b^c$ in each iteration. Consequently, there will be no additional costs for $w_n$, so that the minimal possible length as given in Theorem \ref{thm:reduction} is obtained for $j=n'$.  
		\qed\end{proof}

	\section{Recursive and Asymptotic Results}\label{sec:recursive_asymptotic}
	
	We will now turn our attention to the analysis of Problem \ref{prob:pawnproblem}. The following proposition gives a recursive formula for the solution.
	
	\begin{proposition}~\label{prop:recursion_fc}The function $f_c$ satisfies $f_c(1)=0$ and 
\begin{align*}
f_c(n)&=\min\big\{ f_c(i)+f_c(n-i)+(c+1)n-i \,\big|\, 1 \le i \le n-1 \big\}. 
\end{align*}
	\end{proposition}
	
	\begin{proof}
		In Problem \ref{prob:pawnproblem}, we define chasers and resigners as before. Since we now work on $\mathbb{Z}$, the pawn at $1$ is the lanterne rouge and the pawn at $n$ is the yellow jersey. In total we will need $n-1$ iterations (or we can assume so if $c=0$) by the following simple analog of a special case of Lemma~\ref{claim:n-c}.

		\par\noindent\emph{Claim.} 
			Let $c \in \N$ and $S=[n] \subseteq \Ns$ be the set of pawn positions. Suppose that the pawns merge to one pawn by an optimal set of iterations.
			Then the following holds for each pawn in every iteration.
			\begin{enumerate}[(i)]
				\item If the pawn is a resigner, then it will stay.
				\item If the pawn is a chaser and $c \ne 0$, then it will move.
				\item If the pawn is a chaser and $c = 0$, then we can choose it to move.
			\end{enumerate} 
\par

		This claim can be proved with the same ideas (pawn displacement) Lemma~\ref{claim:n-c} has been proved.
		
		Let $\sigma_j(k)$ be the position after $j$ iterations of the pawn that starts in $k$.
		After $n-2$ iterations, all pawns are merged into the lanterne rouge at $n-1$ and the yellow jersey at $n$.
		Let $I=\sigma_{n-2}^{-1}(n-1)=\{1,2,\ldots,i\}$ (being the peloton) and $J=\sigma_{n-2}^{-1}(n)=\{i+1,\ldots, n\}$ (being the first group). See also Figure ~\ref{fig:optimalraces}.
		
		Now note that the pawn at $i$ is a resigner until the full peloton has merged into one pawn in position $i$. The minimal cost for this is equal to $f_c(i)$. In each of the remaining $n-i$ iterations, this pawn will be a chaser at cost $c+1$.   
		Similarly, the pawn starting in $i+1$ is a chaser until the first group has merged into one pawn in position $n$. This takes $n-i-1$ iterations and the minimal cost to merge the first group is $f_c(n-i)$. In the remaining $i$ iterations, the pawn at $n$ is a resigner at cost $c$. 
		
		So the minimum cost is indeed 
		$f_c(i)+f_c(n-i)+(c+1)n-i$, where we have to minimize over all possible $1\le i \le n-1$.
	\qed\end{proof}
	
	In Figure~\ref{fig:optimalraces} we have presented three ways in which the minimum cost can be attained when $c=1$ and $n=7$ (the pawn at place $3$ having two choices in the right part).  
Here resigners are drawn amber (light) and chasers red (dark). 
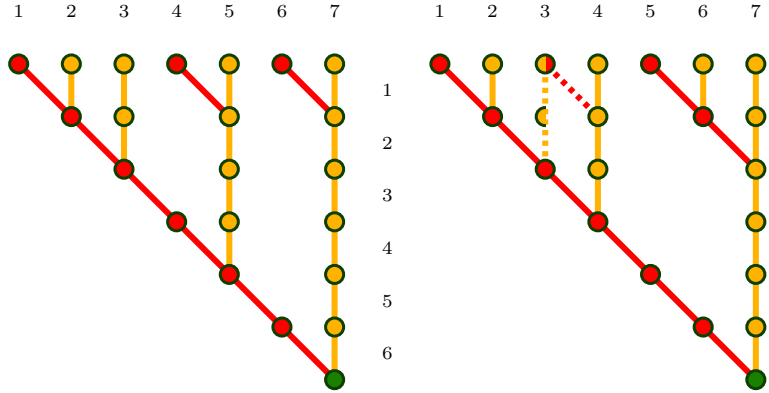
\begin{figure}[ht]
\centering{
\begin{tikzpicture}[x=7mm,y=-7mm]
\foreach \i in {1,...,7} {
  \node at (\i,0) {\scriptsize\i};
  \node at (\i+8,0) {\scriptsize\i};
}
\foreach \s in {1,...,6} {
  \node at (8,\s+0.5) {\scriptsize\s};
}
\draw[myline,Resigner,mydash] (11,1.05) -- (11,3);
\draw[myline,Chaser,mydash] (11.036,1.036) -- (12,2);
\foreach \x/\s in { 1/6, 4/1, 6/1, 9/6, 13/2 } {
   \draw[myline,Chaser] (\x,1) -- (\x+\s,\s+1);
}
\foreach \x/\s in { 2/1, 3/2, 5/4, 7/6, 10/1, 12/3, 14/1, 15/6 } {
   \draw[myline,Resigner] (\x,1) -- (\x,\s+1);
}
\fill[mycirc,Chaser] (11,1) circle[];
\begin{scope}
  \clip (10,0) rectangle (11,3);
  \fill[mycirc,Resigner] (11,1) circle[];
  \draw[mycirc,fill=Resigner] (11,2) circle[];
\end{scope}
\draw[mycirc] (11,1) circle[]; 
\draw[mycirc,fill=lastPawn] (7,7) circle[] (15,7) circle[] (15,7);
\foreach \x/\s in { 1/6, 4/1, 6/1, 9/6, 13/2 } {
   \draw[mycirc,fill=Chaser] foreach \y in {1,...,\s} {
      (\x+\y-1,\y) circle[]
   };
}
\foreach \x/\s in { 2/1, 3/2, 5/4, 7/6, 10/1, 12/3, 14/1, 15/6 } {
   \draw[mycirc,fill=Resigner] foreach \y in {1,...,\s} {
      (\x,\y) circle[]
   };   
}
\end{tikzpicture}
}
\caption{The three optimal races for $n=7$ and $c=1$. The positions are indicated above and the iterations are numbered in the middle. The peloton has size 5 in the left race and size 4 in the other two races.}
\label{fig:optimalraces}
\end{figure}	
In the optimal races in this example, there are either 8 chasers and 13 resigners or 9 chasers and 11 resigners (counted with multiplicity). The total cost $f_1(7)$ therefore is
	\[
	f_1(7) = 8\cdot 2+13\cdot 1 = 9\cdot 2+11\cdot 1= 29.
	\]
	Computing $f_1(1),\ldots,f_1(5)$ by the recursion gives 0, 3, 7, 12, 17 which can be used to alternatively express the total cost $f_1(7)$ by
	\[
	f_1(7) = f_1(5)+f_1(2)+2\cdot 7-5 = f_1(4)+f_1(3)+2\cdot 7-4 = 29.
	\]
In Section \ref{sec:solve_fc}, we will express $f_c(n)$ in terms of recurrent sequences. But this is not necessary to determine the order of growth of the function $f_c(n)$. 
	
	\begin{proposition}\label{prop:bounds_fc}
		For all $c\geq 0$ and $n\geq 1$, we have
		\[
		c n \log_2(n) \le f_c(n)\le (c+\tfrac12) n \big\lceil \log_2(n) \big\rceil.
		\]
	\end{proposition}
	
	\begin{proof}
		Both bounds can be proved by induction, the base case for $n=1$ being true.
		For the lower bound, fix $n$ and assume that $ci \log_2(i)\le f_c(i)$ for $i<n$. 	This implies that for every  $1\le i \le n-1$ we have that  
		\begin{align*}
		f_c(i)+f_c(n-i)+ (c+1)n-i
		&\ge c \big( i \log_2 (i) + (n-i) \log_2 (n-i) + n \big).
		\end{align*}
		Note that $i \log_2 (i) + (n-i) \log_2 (n-i) $ is minimized when $i=\frac n2$ since its derivative (as a function of $i$ for $n$ fixed) is $\log_2(i) - \log_2 (n-i).$ Plugging in $i=\frac n 2$ gives $cn \log_2(n)$ on the right hand side. So, we have $cn \log_2(n)\le f_c(n)$ and we conclude by mathematical induction.
		
		For the upper bound, assuming it is true for values strictly smaller than $n$ (where $n>1$), we have 
		\begin{align*}
		f_c(n) &\le f_c\left( \left \lfloor \frac n2 \right \rfloor \right)+f_c\left(\left  \lceil \frac n2 \right \rceil  \right)+(c+1)n - \left \lceil \frac n2 \right \rceil \\
		&\le \big(c+\tfrac12\big)n \left \lceil \log_2\left( \left \lceil \frac n2 \right\rceil\right) \right \rceil + \big(c+\tfrac12\big)n \\
		&= \big(c+\tfrac12\big) n \big\lceil \log_2(n) \big\rceil.
		\end{align*}
		So again by mathematical induction the bound does hold for every $n$.
	\qed\end{proof}

	As a corollary of Proposition \ref{prop:bounds_fc}, we determine the asymptotic growth of maximal reset thresholds in the \v{C}ern\'y family. 

	\begin{theorem}\label{thm:asymptotics} Denoting the reset threshold of $C_n^c$ by $r(C_n^c)$, we have
		\[
		\max_c r(C_n^c) \sim \tfrac1{4} n^2\log_2(n).
		\]
	
	\end{theorem}

	\begin{proof}
		By Theorem~\ref{thm:reduction} and the fact that $n'(n'-1) + c + 1 \le n^2=o\big(n^2 \log_2(n)\big)$ for $n'= n - c - 1$, it is sufficient to prove that $\max_c f_c(n-c-1) \sim \tfrac1{4} n^2\log_2(n)$. By Proposition~\ref{prop:bounds_fc}, we have for the upper bound that 
		\begin{align*}
f_c(n-c-1) &\le \big(c+\tfrac12\big)(n-c-1) \big\lceil \log_2(n-c-1) \big\rceil \\
&\le \frac{n^2}{4} \big\lceil \log_2(n) \big\rceil = \big(1+o(1)\big)\frac{n^2}{4}\log_2(n).
\end{align*}

		For the lower bound, we choose $c=\left \lfloor \frac {n-1}2 \right \rfloor$, and from the lower bound in Proposition~\ref{prop:bounds_fc} we get that 
		\begin{align*}
		f_c(n-c-1)&\ge c(n-c-1)\log_2(n-c-1)\\ &\ge \frac{(n-1)^2-1}{4}\log_2\Big(\frac n2 - 1\Big) = \big(1-o(1)\big)\frac{n^2}{4}\log_2(n). 
		\end{align*}
		The two bounds together imply the result.
	\qed\end{proof}

One can show that the optimal choice $c'$ for $c$ satisfies $\big|\frac n2 - c'\big| = o(n)$. In Proposition \ref{prop:optimalc} in Section \ref{sec:est_fc}, we will prove that 
$$
0 < \tfrac n2 - c' = \Theta\big(n/{\log(n)}\big).
$$

Another corollary of Proposition \ref{prop:bounds_fc} is the following.

\begin{corollary} \label{betterthancerny}
$r(C^1_n) > r(C^0_n)$ for all $n \ge 6$. So if the \v{C}ern\'y conjecture holds, then $p(n,2)>d(n)$ for all $n\geq 6$.
\end{corollary}

\begin{proof}
If $6 \le n \le 9$, then $r(C^1_n) > r(C^0_n)$ follows from explicit computations. So assume that $n \ge 10$. Then Theorem \ref{thm:reduction} and the lower bound of Proposition \ref{prop:bounds_fc} yields
\begin{align*}
r(C^1_n) &= (n-2)(n-3) + 1 + 1 + f_1(n-2) \\
&\ge (n-2)(n-3) + 1 + 1 + (n-2)\log_2(n-2) \\
&\ge (n-2)(n-3) + 2 + (n-2)3 \\
&= (n-1)^2 + 1 
\end{align*}
Although it is known that $r(C^0_n) = (n-1)^2$, we provide a proof of that in the next section. So $r(C^1_n) > r(C^0_n)$ holds for all $n \ge 6$.
\qed\end{proof}
	
	\section{Explicit Solution of the Pawn Race Problem}\label{sec:solve_fc}
	
	In this section, we determine the solution of Problem \ref{prob:pawnproblem}, from which  by Theorem~\ref{thm:reduction} the exact expressions for the reset thresholds of $C_n^c$ for all $n$ follow as well.
	
	When $c=0$, we see that the lanterne rouge has to move $n-1$ times. If the other pawns do not move, we get the minimum cost of $n-1$, i.e. $f_0(n)=n-1.$ This result can also immediately be derived from Proposition~\ref{prop:recursion_fc}.
	Theorem~\ref{thm:reduction} now gives that the reset threshold of $C_n^0$ is indeed equal to $(n-1)(n-2)+1+(n-2)=(n-1)^2$, which yields an alternative proof for the well-known reset thresholds of the \v{C}ern\'y sequence. 
	
	When $c\ge1$, our approach is based on solving the recursion given in Proposition~\ref{prop:recursion_fc}. In the conference version of this paper \cite{BCD21}, we saw that the Fibonacci numbers play an important role in the solution for $c = 1$, and that the Padovan numbers play a similar role in the solution for $c = 2$. We will see that these numbers enter the picture when determining the set of values of $i$ for which $f_c(i)+f_c(n-i)+(c+1)n-i$ is minimal. 

Let $p : \Ns \rightarrow \Ns$ be increasing and not bounded. Define $m : \Ns \rightarrow \Ns$ by 
$$
m(i) = \min\{j \mid i < p(j)\}.
$$
We call $m$ the \emph{twinverse} (twisted inverse) of $p$.
\begin{figure}[h]
\centering{
\begin{tikzpicture}[x=5mm,y=5mm]
\fill[lightTwins] (0,0) rectangle (1,1);
\fill[lightMichiel] (1,0) rectangle (13,13);
\fill[lightPepijn] (0,1) -- (1,1) foreach \p in { (1,5), (2,5), (2,7), (3,7), (3,8), (4,8), (4,9), (5,9), (5,10), (7,10), (7,11), (9,11), (9,12), (12,12), (12,13) } { -- \p } -- (0,13) -- cycle;
\draw[very thin,black] foreach \x in {1,2,...,13} { (\x,-0.5) -- (\x,13.5) }; 
\draw[very thin] foreach \x in {1,2,...,13} { (\x,-0.5) node[anchor=north] {\scriptsize\x} }; 
\draw[very thin,black] foreach \y in {1,2,...,13} { (-0.5,\y) -- (13.5,\y) };
\draw[very thin] foreach \y in {1,2,...,13} { (-0.5,\y) node[anchor=east] {\scriptsize\y} };
\draw[very thick] (0,13.5) -- (0,0) -- (13.5,0);
\draw[very thick] (0,1) -- (1,1) -- (1,0) (0,13) -- (13,13) -- (13,0);
\draw[myline] (1,1) foreach \p in { (1,5), (2,5), (2,7), (3,7), (3,8), (4,8), (4,9), (5,9), (5,10), (7,10), (7,11), (9,11), (9,12), (12,12), (12,13) } { -- \p} -- (13.5,13);
\draw[mycirc,fill=Pepijn] (1,1) circle[] ++ (0,1) circle[] ++ (0,1) circle[] ++ (0,1) circle[] ++ (1,1) circle[] ++ (0,1) circle[] ++ (1,1) circle[] ++ (1,1) circle[] ++ (1,1) circle[] ++ (2,1) circle[] ++ (2,1) circle[] ++ (3,1) circle[]; 
\draw[mycirc,fill=Michiel] (1,5) circle[] ++ (1,2) circle[] ++ (1,1) circle[] ++ (1,1) circle[] ++ (1,1) circle[] ++ (1,0) circle[] ++ (1,1) circle[] ++ (1,0) circle[] ++ (1,1) circle[] ++ (1,0) circle[] ++ (1,0) circle[] ++ (1,1) circle[] ++ (1,0) circle[];
\draw[mycirc,fill=Pepijn] (15,8) circle[];
\draw[mycirc,fill=Michiel] (15,5) circle[];
\node[anchor=west,inner sep=0pt] at (15.5,8) {$\big\{\big(p(i),i\big) \mid i \in \Ns\big\}$};
\node[anchor=west,inner sep=0pt] at (15.5,5) {$\big\{\big(i,m(i)\big) \mid i \in \Ns\big\}$};
\end{tikzpicture}
}
\caption{Illustration of twinverses.} \label{fig:twins}
\end{figure}
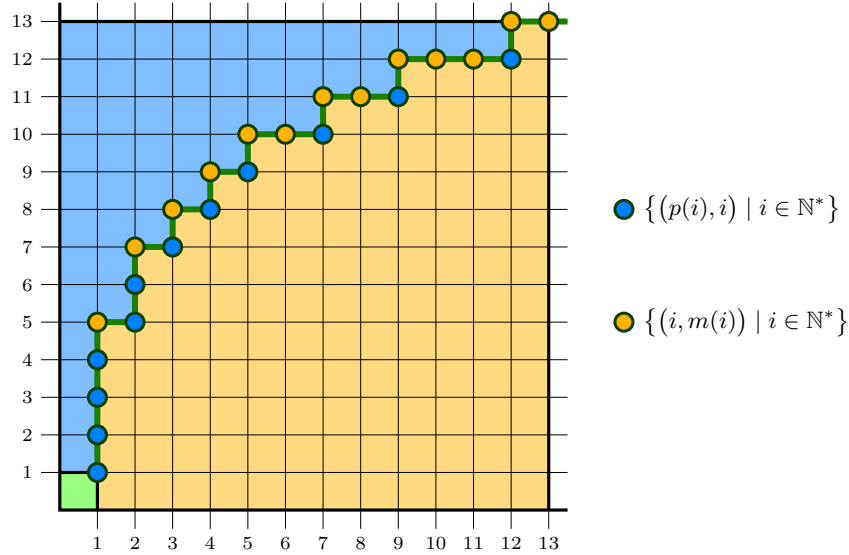

\begin{proposition} \label{prop:twin}
Let $p : \Ns \rightarrow \Ns$ be increasing and not bounded. Then the twinverse of $p$ is increasing and not bounded as well. Furthermore, the twinverse of the twinverse of $p$ is $p$.
\end{proposition}

\begin{proof}
Let $m$ be the twinverse of $p$, and $1 \le i \le i'$. Take $j = m(i)$ and $j' = m(i')$. Then $i \le i' < p(j')$. Since $i < p(j')$, we have $j \le j'$ by definition of $m(i)$. So $m$ is increasing. If $i \ge p(j)$, then $m(i) > j$, so $m$ is not bounded.

Let $p'$ be the twinverse of $m$. Then
$$
p'(i) = \min\{j \mid i < m(j)\}.
$$
So we must show that $i < m(p(i))$, and that $q < p(i)$ implies $i \ge m(q)$.
Indeed
$$
m(p(i)) = \min\{j \mid p(i) < p(j)\} > i,
$$ 
and if $q < p(i)$, then
$$
m(q) = \min\{j \mid q < p(j)\} \le i.
$$
so $p' = p$.
\qed\end{proof}

Lemma \ref{lem:rect} below is illustrated in Figure \ref{fig:twins}, for $m = m_2$, $p = p_2$, and $n = 13$. The definition of $m_2$ and $p_2$ follows later. The total area in Figure \ref{fig:twins} is $n\cdot m(n)=13\cdot 13$ and is the sum of the areas of the three colors.

\begin{lemma} \label{lem:rect}
Let $m$ be the twinverse of $p$. Then
$$
\sum_{i=1}^{m(n)-1} p(i) = n \cdot m(n) - \sum_{j=1}^{n-1} m(j) - 1.
$$
\end{lemma}

\begin{proof}
We show this equality by induction on $n$. To prove the base case $n = 1$, we must show that
$$
\sum_{i=1}^{m(1)-1} p(i) = m(1) - 1.
$$
If $i < m(1)$, then $p(i) = \min\{j \mid i < m(j)\} = 1$. Hence the base case follows.

To prove the induction step, we must show that
$$
\sum_{i=m(n)}^{m(n+1)-1} p(i) = (n+1) \cdot \big(m(n+1) - m(n)\big).
$$
If $m(n) \le i < m(n+1)$, then $p(i) = \min\{j \mid i < m(j)\} = n + 1$. Hence the induction step follows.
\qed\end{proof}

Lemma \ref{lem:imin} below is the key to determining the values of $i$ for which the minimum value in the recursive formula of Proposition \ref{prop:recursion_fc} is reached. This is the most crucial part of the solution of the pawn race problem.

\begin{lemma}\label{lem:imin}
	Let $m$ be the twinverse of $p$. Suppose that 
	$$
	p(k-1)+p(k)\leq n\leq p(k)+p(k+1).
	$$
	Then
	$$
	S(i) := \sum_{j=1}^{n-i-1} m(j) + \sum_{j=1}^{i-1} m(j) - i,\qquad i\in \{1,2,\ldots,n-1\}
	$$
	is minimal at $i$, if and only if
	\begin{equation} \label{ieq}
	p(k-1) \le n-i \le p(k) \le i \le p(k+1).
	\end{equation}
\end{lemma}

\begin{proof} 
	Notice that 
	$$
	\Delta S(i) := S(i+1) - S(i) = m(i) - m(n-1-i) - 1.
	$$
We distinguish three cases for $i$, in order of increasing $i$. 
\begin{itemize}
	
	\item $i < p(k)$ or $p(k) < n - i$.
	\begin{itemize}
		\item[$*$] If $i < p(k)$, then $p(k-1) \le n - p(k) < n - i$, so $m(i) \le k \le m(n-i-1)$.
		\item[$*$] If $p(k) < n - i$, then $i < n - p(k) \le p(k+1)$, so $m(i) \le k + 1 \le m(n-i-1)$.
	\end{itemize}
	In both cases, $\Delta S(i) < 0$, so there is no minimum at $i$. 
	
	\item $p(k-1) < n-i \le p(k) \le i < p(k+1)$.\\
	Then $m(i) = k+1$ and $m(n-i-1) = k $. So $\Delta S(i) = 0$.
	
	\item $n-i \le p(k-1) \le p(k) \le i$ or $n-i \le p(k) \le p(k+1) \le i$.
	\begin{itemize}
		\item[$*$] If $n-i \le p(k-1) \le p(k) \le i$, then $m(i) \ge k+1$ and $m(n-i-1) \le k - 1$.
		\item[$*$] If $n-i \le p(k) \le p(k+1) \le i$, then $m(i) \ge k+2$ and $m(n-i-1) \le k$.
	\end{itemize}
	In both cases, $\Delta S(i) > 0$, so there is no minimum at $i+1$. 	
\end{itemize}
We conclude that the function $S(i)$ is first decreasing, then possibly constant and then increasing. The function is minimal if and only if (\ref{ieq}) is satisfied.
\end{proof}
	
Notice that $p(k) \in \{n-i,i\}$ for the smallest $i$ which satisfies \eqref{ieq}. So
$i = \max\{p(k),n-p(k)\}$ is the smallest solution of \eqref{ieq}.

Let $c\in\Ns$. Our next result will express the solution $f_c(n)$ of Problem \ref{prob:pawnproblem} in terms of the recursive sequence $p_c(k), k\in\Ns$, defined by
\[
p_c(k) = \begin{cases} 1 &\text{if}\ 1\leq k\leq 2c,\\p_c(k-c-1)+p_c(k-c) \qquad & \text{if}\ k\geq 2c+1.
\end{cases}
\]
Let $m_c$ be the twinverse of $p_c$ for all $c$. The sequence $p_1$ are the Fibonacci numbers. With the exception of $p_2(1)$ (which does not follow the recurrence formula), the sequence $p_2$ is a shift of the Padovan numbers. Define the sequence $q_c(k)$ for $k\geq 1$ by
\begin{equation} \label{qeqc}
q_c(k) = 1+\sum_{i=1}^{k-1} p_c(i).
\end{equation}
Then $q_c(k) = k$ for $k\leq 2c$. For $k\geq 2c+1$, the definition of $p_c$ gives the recursion
\begin{align*}
q_c(k) &= 1+2c+\sum_{i=2c+1}^{k-1}\big(p_c(i-c-1)+p_c(i-c)\big)\\
&= 1+2c+\sum_{i=c}^{k-c-2}p_c(i)+\sum_{i=c+1}^{k-c-1}p_c(i) = q_c(k-c-1)+q_c(k-c).
\end{align*}
From this recursion, we infer that
$$
q_c(k) = q_c(k+c+1) - q_c(k+1) = \sum_{i=k+1}^{k+c} p_c(i) \qquad (k \ge c).
$$
In particular,
\begin{alignat}{2}
q_1(k) &= p_1(k+1) \qquad && (k \ge 1) \label{qeqc1} \\
q_2(k) &= p_2(k+1) + p_2(k+2) = p_2(k+4) \qquad && (k \ge 2). \label{qeqc2}
\end{alignat}
We are now ready to formulate the main result of this section.


%

\begin{theorem} \label{thm:pawnrace}
Suppose that $c \in \Ns$. Then
\begin{equation} \label{fsol}
f_c(n) = \sum_{j=1}^{n-1} m_c(j) = n \cdot m_c(n) - q_c\big(m_c(n)\big).
\end{equation}
\end{theorem}

We illustrate Theorem \ref{thm:pawnrace} for $n=7$ and $c=1$. Since $p_1(5) = 5$ and $p_1(6)=8$, we find $m_1(7)=6$. Furthermore, $q_1(6) = p_1(7) = 13$ on account of \eqref{qeqc1}, and therefore $f_1(7) = 7\cdot 6-13 = 29$, in agreement with the example in the previous section.
	
\ifijfcs	
	\begin{proof}[Proof of Theorem \ref{thm:pawnrace}]
\else
	\begin{proof}[of Theorem \ref{thm:pawnrace}]
\fi	
	The second equality of \eqref{fsol} follows from Lemma \ref{lem:rect} and \eqref{qeqc}. So it remains to prove the first equality to obtain \eqref{fsol}.
	
	From $p_c(2c) = 1$ and $p_c(2c+1) = 2$, we infer that $m_c(1) = 2c+1$. This yields the cases $n = 1$ and $n = 2$. So assume that $n \ge 3$. We prove \eqref{fsol} by induction on $n$, so we assume that 
	\begin{align}\label{eq:inductionhyp}
	f_c(i) = \sum_{j=1}^{i-1} m_c(j) \qquad (i < n).
	\end{align}
	Choose $k$ such that 
	\begin{align}\label{eq:k_choice}
	p_c(k+c) < n \le p_c(k+c+1).
	\end{align}
	Then $m_c(n-1) = k+c+1$. From $p_c(2c+1) = 2 < n$, we infer that $k + c + 1 > 2c + 1$, so 
	$k - 1 \ge c$. Hence 
	\begin{align}\label{eq:k_choice2}
	p_c(k+c) &= p_c(k-1) + p_c(k), & p_c(k+c+1) &= p_c(k) + p_c(k+1).
	\end{align}
	By (\ref{eq:k_choice}) and (\ref{eq:k_choice2}), the condition of Lemma \ref{lem:imin} is satisfied, so we can choose $i$ such that
	\[
	p_c(k-1) \leq n-i\leq p_c(k) \leq i \leq p_c(k+1).
	\]
	By Proposition \ref{prop:recursion_fc}, the induction hypothesis (\ref{eq:inductionhyp}) and Lemma \ref{lem:imin}, we find
	\begin{align}\label{eq:fcn}
	f_c(n) &= f_c(n-i) + f_c(i) + (c+1)n - i.
	\end{align}
	Below, we will find similar formulas for $f_c(n-1)$ instead of $f_c(n)$. As $n > p_c(k+c)$, at least one of two cases applies:
	\begin{itemize}
		
		\item \emph{Case $i > p_c(k)$.} In this case
		$$
		p_c(k-1) \le n - i \le p_c(k) \le i - 1 < p_c(k+1).
		$$
		Hence $f_c(n-1) = f_c(n-i) + f_c(i-1) + (c+1)(n-1) - (i - 1)$. So
		$$
		f_c(n) - f_c(n-1) = m_c(i-1) + (c+1) - 1 = (k + 1) + c = m_c(n-1).
		$$
		
		\item \emph{Case $n - i > p_c(k-1)$.} In this case 
		$$
		p_c(k-1) \le n - i - 1 < p_c(k) \le i \le p_c(k+1).
		$$
		Hence $f_c(n-1) = f_c(n-i-1) + f_c(i) + (c+1)(n-1) - i$. So
		$$
		f_c(n) - f_c(n-1) = m_c(n-i-1) + (c+1) = k + (1 + c) = m_c(n-1).
		$$
		
	\end{itemize}
	This completes the proof of the first equality in \eqref{fsol}. 
	\end{proof}

The next result characterizes the number of optimal solutions of the pawn race. A unique optimal solution exists if and only if $n$ is an element of the sequence $p_c$. 

\begin{theorem}\label{thm:pawnrace2}
	Let $c\geq 1$ and denote the number of optimal solutions of the pawn race by $o_c(n)$. Let $k = m_c(n)-c-1$ and define
	\[
	I_n = \{i \,|\, p_c(k-1)\leq n-i\leq p_c(k)\leq i\leq p_c(k+1)\}.
	\]
	Then $o_c(1)=o_c(2) = 1$ and for $n\geq 3$, 
	\[
	o_c(n) = \sum_{i\in I_n} o_c(n-i)o_c(i).
	\]
	Furthermore, 
	$
	o_c(n) = 1 \Longleftrightarrow n = p_c(k')$ for some $k' \Longleftrightarrow n = p_c\big(m_c(n) - 1\big).
	$
\end{theorem}

We illustrate the result for $n=7$ and $c=1$. Since $7$ is not a Fibonacci number, there are multiple solutions, see Figure \ref{fig:optimalraces}. These are precisely the optimal races. Notice that $k = m_1(7) - 1 - 1 = 4$, so $p_1(k-1) = 2$, $p_1(k) = 3$ and $p_1(k+1) = 5$. Hence $2 \le 7 - i \le 3 \le i \le 5$. This is exactly the case if either $i = 4$ or $i = 5$, so 
$$
o_c(7) = o_c(2)o_c(5) + o_c(3)o_c(4).
$$
Since both $2$ and $5$ are Fibonacci numbers, $o_c(2)o_c(5) = 1\cdot1 = 1$. The race which corresponds to the term $o_c(2)o_c(5) = 1$ is illustrated on the left hand side of Figure \ref{fig:optimalraces}. The size of the first group and the peloton are $2$ and $5$ respectively, corresponding to the arguments of $o_c$. For $o_c(3)$ and $o_c(4)$, we have $o_c(3) = o_c(1)o_c(2) = 1$ and 
$$
o_c(4) = o_c(1)o_c(3) + o_c(2)o_c(2) = 2,
$$
so $o_c(3)o_c(4) = 1\cdot2 = 2$. The races which correspond to the term $o_c(3)o_c(4) = 2$ are illustrated on the right hand side of Figure \ref{fig:optimalraces}. The dashes indicate the two distinct optimal solutions for the pawn race of the peloton of size $4$.

\ifijfcs
	\begin{proof}[Proof of Theorem \ref{thm:pawnrace2}]
\else
	\begin{proof}[of Theorem \ref{thm:pawnrace2}]
\fi
	For $n\leq 2$, all three statements in the last line of the theorem are true. For $n\geq 3$, we distinguish two cases.
	\begin{itemize}
		
		\item \emph{Case $n = p_c(k')$ for some $k'$.}
		
		Take $k'$ maximal as such. Then $k' = m_c(n)-1$, which proves the last equivalence. Furthermore, $k' = k+c$, so $n = p_c(k)+p_c(k-1)$. Hence Lemma \ref{lem:imin} applies, and $I_n = \{p_c(k)\}$ follows.		
		
		\item \emph{Case $n \ne p_c(k')$ for any $k'$.}
		
		Again, let $k' = m_c(n)-1$. Then $p_c(k') < n < p_c(k'+1)$ and $k' = k+c$, so $p_c(k-1)+p_c(k) < n <p_c(k)+p_c(k+1)$. In particular, Lemma \ref{lem:imin} applies and $p_c(k+1) > p_c(k-1) + 1$. But by means of the recurrence relation of $p_c$, one can show by induction that $p_c(k+1) \le p_c(k-1) + 1$ if either $p_c(k-1) = p_c(k)$ or $p_c(k) = p_c(k+1)$. Consequently,
$$
p_c(k-1) < p_c(k) < n - p_c(k-1) \qquad \mbox{and} \qquad n - p_c(k+1) < p_c(k) < p_c(k+1).
$$	
		
		To prove that there are at least two solutions of \eqref{ieq}, suppose that $i$ is a solution and $i+1$ is not. Then either $n-i=p_c(k-1)$ or $i = p_c(k+1)$. In both cases, $n-i<p_c(k)<i$, so $i-1$ is another solution and $\{i-1,i\} \subseteq I_n$.
		\end{itemize}
In both cases, the inductive formula of $o_c(n)$ holds. That the solution is unique if and only if $n = p_c(k')$ for some $k'$ follows by induction as well.\qed
\end{proof}

The function $f_c$ is affinely linear between consecutive values of $p_c$, i.e.\@
\begin{equation}
f_c\big(\lambda p_c(k) + (1 - \lambda) p_c(k+1)\big) = \lambda f_c\big(p_c(k)\big) + (1 - \lambda) f_c\big(p_c(k+1)\big)
\end{equation}
for all $\lambda \in [0,1]$ for which the left hand side makes sense. This is because $m_c(n) = k+1$ for all $n$ for which $p_c(k) \le n < p_c(k+1)$.
Since $p_c(k+1)-p_c(k)$ can be arbitrary large, $f_c$ cannot be represented by a polynomial.

Combining the results in the current section with Theorem \ref{thm:reduction} gives expressions for the reset thresholds of all automata in the \v{C}ern\'y family. Asymptotic estimates are given in the next section.

\begin{corollary}
	Suppose that $1\leq c \leq n-2$ and denote $n' = n - c - 1$. 
	If $w$ is a shortest synchronizing word for $C_n^c$, then
	\begin{align*}
	|w| = r(C^c_n)
	&= n'(n'-1) + c+1 + n' m_c(n') - q_c\big(m_c(n')\big).
	\end{align*}
	Furthermore, $w$ is unique, if and only if $n' = p_c(k')$ for some $k'$.
\end{corollary}

\section{Estimates of the Pawn Race Problem} \label{sec:est_fc}

To estimate $f_c$ asymptotically for $c \ne 0$, we need to take a look at the characteristic polynomial $\chi_c(x) =  x^{c+1} - x - 1$ of the recurrence relation of $p_c$ and $q_c$. Since $\chi_c(1) = -1 < 0<\chi_c(2) $ and $\chi_c$ is strictly increasing for $x\ge1$, we deduce that $\chi_c(x) = 0$ for a unique real number $x > 1$. Let $\phi_c$ be this number. Note that $\phi_c$ decreases and approaches 1 as $c$ increases. In Proposition \ref{prop:bounds_fc} we already determined the order of growth of $f_c(n)$. For fixed $c$, Proposition \ref{prop:fest} below gives a more precise result and yields an asymptotic estimate for $n \rightarrow \infty$. 

\begin{proposition}\label{prop:fest}
For all $c\geq 1$ and $n\geq 1$, we have
$$
\frac{n\log(n)}{\log(\phi_c)} - 3cn < f_c(n) < \frac{n\log(n)}{\log(\phi_c)} + (c+1)n. 
$$
\end{proposition}

\begin{proof} By mathematical induction we obtain that
\begin{equation}
\phi_c^{k-2c} \le p_c(k) \le \phi_c^{k-c} \qquad (k \ge c). \label{pest}
\end{equation}
Since $m_c(n) - 1 \ge 2c$ and $p_c\big(m_c(n)-1\big) \le n < p_c\big(m_c(n)\big)$, we infer from \eqref{pest} that
\begin{equation} \label{nest}
\phi_c^{m_c(n)-2c-1} \le n < \phi_c^{m_c(n)-c} \qquad (n \ge 1),
\end{equation}
so
\begin{equation}
\frac{\log(n)}{\log(\phi_c)} + c < m_c(n) \le \frac{\log(n)}{\log(\phi_c)} + 2c + 1 \qquad (n \ge 1). \label{mest}
\end{equation}
If $x > 1$, then $2x > x+1 > 2\sqrt{x}$, so $2\phi_c > \phi_c^{c+1} > 2\sqrt{\phi_c}$ and $\phi_c^c < 2 < \phi_c^{c+1/2}$. 
Since $1 < \phi_c^c < 2$, one can prove by mathematical induction that
\begin{equation} 
c\,\phi_c^{k-c} \le q_c(k) \le 2c\,\phi_c^{k-c-1} \qquad (k \ge c). \label{qest}
\end{equation}
If we combine this with \eqref{nest}, then we obtain $cn < q_c\big(m_c(n)\big) \le 2cn\,\phi_c^c$, so
\begin{equation}
cn < q_c\big(m_c(n)\big) < 4cn \qquad (n \ge 1). \label{qmest}
\end{equation}
Applying the bounds in \eqref{mest} and \eqref{qmest} to Theorem \ref{thm:pawnrace} completes the proof.
\qed\end{proof}

As a consequence, we now obtain asymptotic estimates for the \v{C}ern\'y family.

\begin{corollary}
	For all $c_n$ such that $1 \leq c_n\leq n-2$,
	\begin{align*}
	r(C^{c_n}_n) &= n^2 + \frac{(n-c_n)\log(n)}{\log(\phi_{c_n})}-c_n\cdot O(n).
	\end{align*}
\end{corollary}

\begin{proof}
In the proof of Proposition \ref{prop:fest}, we saw that $c<\log(2)/\log(\phi_c)<c+\frac12$. Consequently,
$$
 \frac{(n-c_n-1)\log(n-c_n-1)}{\log(\phi_{c_n})} = \frac{(n-c_n)\log(n)}{\log(\phi_{c_n})}
- c_n \cdot O(n).
$$
Now it is easy to obtain the result from Theorem \ref{thm:reduction} and Proposition \ref{prop:fest}.
\qed\end{proof}

The reset thresholds $r(C_n^c)$ provide lower bounds for $p(n,2)$. To get the best lower bounds, one should maximize over $c$. Just as from the proof of Theorem \ref{thm:asymptotics}, one can infer from the above that $\big|\frac{n}2 - c'\big| = o(n)$ for optimal values $c'$ of $c$. We 
\ifijfcs
end
\else
continue 
\fi
this section with proving that $0 < \frac{n}2 - c' = \Theta\big(n/{\log(n)}\big)$, where we only assume that $c'$ is a local maximum.
 
 First we need a small lemma to compare solutions of the pawn race problem.

\begin{lemma}\label{lem:comparec}
	If $n' < c$, then $f_{c-1}(n'+1) > f_{c}(n')$.
\end{lemma}

\begin{proof}
If we replace $c$ by $c+1$ in Problem \ref{prob:pawnproblem}, then the price of moving becomes $1 + 1/(c+1)$ times larger, and the price of staying becomes $1 + 1/c$ times larger. Consequently,
\begin{equation}
1 + \frac1{c+1} < \frac{f_{c+1}(n')}{f_{c}(n')} < 1 + \frac1c \qquad (n' \ge 2). \label{fc}
\end{equation}
On account of Theorem \ref{thm:pawnrace},
$$
\frac{n'\cdot m_c(n')}{f_c(n')} = \frac{f_c(n')+q_c\big(m_c(n')\big)}{f_c(n')} = 
1 + \frac{q_c\big(m_c(n')\big)}{f_c(n')} \qquad (n' \ge 2),
$$
and
\begin{equation} \label{fn}
\frac{f_c(n'+1)}{f_c(n')} = 1 + \frac1{n'} \bigg(1 + \frac{q_c\big(m_c(n')\big)}{f_c(n')}\bigg)\qquad (n' \ge 2).
\end{equation}
Since $f_{c-1}(2) > 0 = f_{c}(1)$, the case $n' = 1$ follows. The case $n' \ge 2$ follows from
$$
\frac{f_{c-1}(n'+1)}{f_{c-1}(n')}  > 1 + \frac1{n'} \ge 1 + \frac1{c-1} > \frac{f_{c}(n')}{f_{c-1}(n')}
$$
which is an application of \eqref{fn} and \eqref{fc}.
\qed\end{proof}

\begin{proposition} \label{prop:optimalc}	
	If $c = c'$ is a local maximum of $r(C^c_n)$, then $c' < \frac n2$ and $c' = \frac n2 - \Theta\big(n/{\log(n)}\big)$.
\end{proposition}

\begin{proof}	
	To prove the first part, suppose that $c' \ge \frac n2$, and take $n' = n - c '- 1$. Then $c' > n'$, so by Lemma \ref{lem:comparec}, $f_{c'-1}(n'+1) > f_{c'}(n')$. Hence $c = c'$ is not a local maximum of $f_c(n-c-1)$. The sum of the other terms of $r(C^c_n)$ only makes things worse, because
	$$
	c' + (n'+1)n' = c' - 2n' + n'(n'-1) > c' + 1 + n'(n'-1).
	$$
	So $c=c'$ is not a local maximum of $r(C^c_n)$.
	
	So it remains to prove the second part.  Suppose that $c = c'$ a local maximum of $r(C^c_n)$, and $n \ge 6$. Then $0 < c' < \frac n2 \le n - 3$ on account of Corollary \ref{betterthancerny}, so $1 \le c' \le n-4$. Since $c = c'$ a local maximum of $r(C^c_n)$, it follows that
	$$
	r(C^{c'-1}_n) \le r(C^{c'}_n) \ge r(C^{c'+1}_n)
	$$
	Let $n' = n - c' - 1$. Using \eqref{fn} and \eqref{fc}, we infer from $r(C^{c'}_n) \ge r(C^{c'+1}_n)$ that
	\begin{align*}
	0 &\le \frac{(c'+1)(n'-1)}{f_{c'}(n'-1)}\big(r(C^{c'}_n) - r(C^{c'+1}_n)\big) \\
	&= \frac{(c'+1)(n'-1)}{f_{c'}(n'-1)}\big(f_{c'}(n') - f_{c'+1}(n'-1) - 1 + 2(n'-1)\big) \\
	&\le (c'+1) - (n'-1) + \frac{(c'+1)q_{c'}\big(m_{c'}(n'-1)\big) + \Theta\big(c'(n')^2\big)}{f_{c'}(n'-1)}
	\end{align*}
	On account of Proposition \ref{prop:bounds_fc} and \eqref{qmest},
	$$
	n' - c' = 2 + \frac{O\big(c'n'n\big)}{\Theta\big(c'n'{\log(n')}\big)} = O\big(n/{\log(n)}\big)
	$$
	because $n' > \frac n2 - 1$. So we may assume that $c' \ge 2$. Just like $n' - c' = O\big(n/{\log(n)}\big)$ was obtained from $r(C^{c'}_n) \ge r(C^{c'+1}_n)$, $n' - c' = \Omega\big(n/{\log(n)}\big)$ can be obtained from $r(C^{c'-1}_n) \le r(C^{c'}_n)$. So $n' - c'=\Theta\big(n/{\log(n)}\big)$ and $c' = \frac n2 - \Theta\big(n/{\log(n)}\big)$.
	\qed\end{proof}

\ifijfcs
It is easier to compute $f_c(n)$ with the explicit solution of Theorem \ref{thm:pawnrace} than with the recursive formula of Proposition \ref{prop:recursion_fc}. We did that for all $n,c \le 10000$, to compute $r(C^c_n)$ for all $n \le 10000$ and $c \le n-2$. It appeared that the optimal choice $c=c'$ does not always increase regularly along with $n$. Most of the times, $c'$ stays the same or increases $1$ if $n$ increases $1$. But sometimes, $c'$ drops significantly if $n$ increases $1$. This odd behaviour is discussed in some more detail in the arXiv-version of this paper \cite{arXiv}.

\else

We finally give a more accurate estimate of $f_c(n)$, namely with an error of $o(n)$. This estimate has not been included in the journal paper. We provide Binet's formulas for $p_c$ and $q_c$ to obtain the new estimate of $f_c(n)$. But first, we need a proposition about the characteristic polynomial.

\begin{proposition}
Let $c \ge 1$. The characteristic polynomial $\chi_c(x) = x^{c+1} - x - 1$ has $c + 1$ distinct roots. Furthermore, all roots $\alpha$ except $\phi_c$ satisfy $|\alpha| < \phi_c$.
\end{proposition}

\begin{proof}
Suppose that $\alpha$ is a double root of $\chi_c(x)$. Then $\alpha$ is also a root of $\chi_c'(x) = (c+1) x^c - 1$, so $|\alpha| \le 1$ and $\alpha$ is a root of $x \chi_c'(x)  - (c+1) \chi_c(x) = c x + (c+1)$. Hence $\alpha = -(c+1)/c$. This contradicts $|\alpha| \le 1$.

Suppose that $\alpha$ is a root of $\chi_c(x)$ such that $|\alpha| \ge \phi_c$. Then $|1 + \alpha^{-1}| \le 1 + \phi_c^{-1}$, and for equality $\alpha$ needs to be $\phi_c$. Furthermore
$$
|1 + \alpha^{-1}| = |\alpha^c| \ge \phi_c^c = 1 + \phi_c^{-1},
$$
so $\alpha = \phi_c$ is the only possibility. \qed
\end{proof}

The following theorem and its proof was inspired by \cite{lee}.

\begin{theorem}
Let $\lambda_1, \lambda_2, \ldots, \lambda_{c+1}$ be the distinct roots of $\chi_c(x) = x^{c+1} - x - 1$. Then 
\begin{align*}
p_c(k) &= \sum_{j=1}^{c+1}\frac{\lambda_j^2}{(\lambda_j+1)(c\,\lambda_j+c+1)(\lambda_j-1)} \lambda_j^k \qquad (k \ge c), \\
q_c(k) &= \sum_{j=1}^{c+1}\frac{\lambda_j^2}{(\lambda_j+1)(c\,\lambda_j+c+1)(\lambda_j-1)^2} \lambda_j^k \qquad (k \ge c). 
\end{align*}
Furthermore, the summands with $\lambda_j = \phi_c$ of the indexed sums are asymptotics of $p_c(k)$ and $q_c(k)$ respectively.
\end{theorem}

\begin{proof}
The last claim follows from the fact that $\phi_c$ is larger than the absolute value of any other root of $\chi_c(x)$.

Notice that 
$$
\sum_{i=k+1}^{k+c} \lambda_j^i = \bigg( \sum_{i=1}^{c} \lambda_j^i \bigg) \lambda_j^k = \frac{\lambda_j^{c+1}-\lambda_j}{\lambda_j-1} \lambda_j^k
= \frac{1}{\lambda_j-1} \lambda_j^k.
$$
So the equality for $q_c(k)$ follows from that for $p_c(k)$ by way of the displayed equality which precedes \eqref{qeqc1} and \eqref{qeqc2}. To prove the equality for $p_c(k)$, we must show that
$$
p_c(k+c) = \sum_{j=1}^{c+1}\frac{\lambda_j^{c+2}}{(\lambda_j+1)(c\,\lambda_j+c+1)(\lambda_j-1)} \lambda_j^k \qquad (k \ge 0).
$$
It suffices to show this for $0 \le k \le c$ only, since the recurrence formula gives the equality for larger $k$. As $p_c(k+c) = 1$ for $0 \le k \le c$ and $\lambda_j^{c+2} = \lambda_j(\lambda_j+1)$, we must show that
$$
\sum_{j=1}^{c+1}\frac{\lambda_j}{(c\,\lambda_j+c+1)(\lambda_j-1)} \lambda_j^k = 1 \qquad (0 \le k \le c).
$$

If we solve
$$
\left( \begin{array}{cccc}
1 & 1 & \cdots & 1 \\
\lambda_1 & \lambda_2 & \cdots & \lambda_{c+1} \\
\vdots & \vdots & \ddots & \vdots \\
\lambda_1^c & \lambda_2^c & \cdots & \lambda_{c+1}^c
\end{array} \right) \left( \begin{array}{c}
y_1 \\ y_2 \\ \vdots \\ y_{c+1}
\end{array} \right) = \left( \begin{array}{c}
1 \\ 1 \\ \vdots \\ 1
\end{array} \right)
$$
with Cramer's rule, then we obtain $y_j = \det(\Lambda_j)/\det(\Lambda)$, where $\Lambda$ is the Vandermonde matrix on the left hand side, and $\Lambda_j$ is obtained from $V$ by replacing the $j$-th column by the right hand side. Comparing the Vandermonde determinants $\det(\Lambda_j)$ and $\det(\Lambda)$ yields 
\begin{align*}
\frac{\det(\Lambda_j)}{\det(\Lambda)} = \prod_{i \ne j}\frac{1-\lambda_i}{\lambda_j-\lambda_i} = \frac{\frac{\chi(1)}{1-\lambda_j}}{\chi'(\lambda_j)} &= \frac{1}{\big((c+1)\lambda_j^c-1\big)(\lambda_j-1)} \\
&= \frac{\lambda_j}{(c\,\lambda_j+c+1)(\lambda_j-1)},
\end{align*}
which gives the required equality. \qed
\end{proof}

\begin{corollary} \label{corbinet}
For $n$ of the form $n = p_c(k)$,
$$
f_c(n) = \bigg(\frac{\ln(n)}{\ln (\phi_c)} + \frac{\ln \big((\phi_c + 1)(c\,\phi_c + c + 1)(\phi_c - 1)\big)}{\ln (\phi_c)} - 2 - \frac{1}{\phi_c-1}\bigg) n \pm o(n).
$$
\end{corollary}

\begin{proof}
Suppose that $n = p_c(k)$. Notice that 
$$
n = p_c(k) \sim \frac{\phi_c^{k+2}}{(\phi_c+1)(c\,\phi_c+c+1)(\phi_c-1)}
$$
so
$$
k + 2 = \frac{\ln(n) + \ln \big((\phi_c + 1)(c\,\phi_c + c + 1)(\phi_c - 1)\big)}{\ln (\phi_c)} \pm o(1).
$$
From $n = p_c(k)$, it follows that $m_c(n) = k + 1$. So
\begin{align*}
m_c(n) &= \frac{\ln(n)}{\ln (\phi_c)} + \frac{\ln \big((\phi_c + 1)(c\,\phi_c + c + 1)(\phi_c - 1)\big)}{\ln (\phi_c)} - 1 \pm o(1). \\
\intertext{Furthermore,}
q_c(m_c(n)) &= q_c(k+1) \sim \phi_c \cdot q_c(k) \sim \frac{\phi_c}{\phi_c-1} \cdot p_c(k) = \bigg(\frac{1}{\phi_c-1} + 1\bigg) \cdot n.
\end{align*}
Now the result follows from Theorem \ref{thm:pawnrace}.
\end{proof}

For $n$ between $p_c(k)$ and $p_c(k+1)$, the value of $f_c(n)$ can be obtained by way of linear interpolation, because $f_c$ is linear between $p_c(k)$ and $p_c(k+1)$. This linear interpolation can also be applied to the asymptotic formula for $f_c(n)$ in Corollary \ref{corbinet}, to extend Corollary \ref{corbinet} to all $n$. This linear interpolation gives larger values than the formula itself, because the formula is a convex function. The reader may show that the values are $\Theta(cn)$ larger on average. So there does not seem to be a nice asymptotic formula for $f_c(n)$ with error $o(n)$.

\addtocounter{section}{-1}
\renewcommand{\thesection}{\arabic{section}A}
\section{Drops in the optimal value of $c$} \label{sec:drops}
\renewcommand{\thesection}{\arabic{section}}

It is easier to compute $f_c(n)$ with the explicit solution of Theorem \ref{thm:pawnrace} than with the recursive formula of Proposition \ref{prop:recursion_fc}. We did that for all $n,c \le 10000$, to compute $r(C^c_n)$ for all $n \le 10000$ and $c \le n-2$. It appeared that the optimal choice $c=c'$ does not always increase regularly along with $n$. Most of the times, $c'$ stays the same or increases $1$ if $n$ increases $1$. But sometimes, $c'$ drops significantly if $n$ increases $1$.

The reason for the drops of $c'$ is as follows. Since $p_c$ has the value $1$ $2c$ times in succession, one can prove by induction that $p_c$ has the value $2^t$ $c + 1 - t$ times in succession if $1 \le t \le c$. Consequently
$$
m_c(2^t) - m_c(2^t-1) = c + 1 - t \qquad (1 \le t \le c)
$$
For small values of $t$, this can be a large leap. The values of $\ldots, m_c(2^t - 2), m_c(2^t - 1)$ are relatively small, and the values of $m_c(2^t), m_c(2^t + 1), \cdots$ are relatively large. To obtain a large value of $f_c(n) = \sum_{j=1}^{n-1} m_c(j)$, it is better for the last few summands to be relatively large than to be relatively small. The drop of $c'$ is a transition from avoiding relatively small last summands to adopting relatively large last summands. The drops occur when $n' = n - c' - 1$ is close to a power of $2$.

The first drop of $c'$ is between $n = 47$ and $n = 48$: $c'$ drops from $15$ to $14$. The next drop of $c'$ is at $n = 99$: $c'$ drops from $35$ to $33$. For $n = 99$, $c'$ has $2$ optimal values that are more than $1$ apart. This does not occur for other $n \le 10000$. $c'$ indeed drops from $35$ to $33$ at $n = 99$, because the values of $c'$ at $n = 98$ and $n = 100$ are $35$ and $33$ respectively.

So the value $c'=35$ at $n = 99$ can be seen as the continuation of the value $c'=35$ at $n = 98$. The value $c'=33$ at $n = 99$ does not come entirely out of the blue, because $c = 32$ is a local optimal choice for $r(C^c_n)$ at $n = 98$. The value $c'=33$ at $n = 99$ continues as the value $c'=33$ at $n = 100$. $c = 36$ is a local optimal choice for $r(C^c_n)$ at $n = 100$, so it can be seen as the continuation of the value $c'=35$ at $n = 99$.

\begin{figure}[ht]
	\colorlet{goptc}{Twins}
	\colorlet{loptc}{Michiel}
	\centering{
		\begin{tikzpicture}[x=0.7mm,y=0.7mm,radius=0.4mm,very thick]
		\begin{scope}[shift={(0,675)},very thin]
		\draw (1656,0) -- (1813,0);
		\draw (1709,-1) -- (1709,1)  (1738,-1) -- (1738,1) (1737,-1) -- (1737,1) (1768,-1) -- (1768,1);
		\draw[anchor=north,inner sep=0pt]  (1709,-2) node {1709};
		\draw[anchor=north west,inner sep=0pt]  (1738,-2) node {1738};
		\draw[anchor=north east,inner sep=0pt]  (1737,-2) node {1737};
		\draw[anchor=north,inner sep=0pt]  (1768,-2) node {1768};
		\draw(1800,-4) node {$n$};
		\end{scope}
		\begin{scope}[shift={(1656,0)},very thin]
		\draw (0,675) -- (0,747);
		\draw (-1,683) -- (1,683)  (-1,696) -- (1,696)  (-1,730) -- (1,730)  (-1,744) -- (1,744);
		\draw[anchor=east,inner sep=0pt] (-2,683) node {683};
		\draw[anchor=east,inner sep=0pt] (-2,696) node {696};
		\draw[anchor=east,inner sep=0pt] (-2,730) node {730};
		\draw[anchor=east,inner sep=0pt] (-2,744) node {744};
		\draw(-4,713) node {$c$};
		\end{scope}
		\draw[loptc] (1768,744) -- (1767,743) -- (1766,743) -- (1765,742) -- (1764,742) -- (1763,741) -- (1762,741) -- (1761,740) -- (1760,740) -- (1759,740) -- (1758,739) -- (1757,739) -- (1756,738) -- (1755,738) -- (1754,737) -- (1753,737) -- (1752,736) -- (1751,736) -- (1750,736) -- (1749,735) -- (1748,735) -- (1747,734) -- (1746,734) -- (1745,733) -- (1744,733) -- (1743,732) -- (1742,732) -- (1741,732) -- (1741,731) -- (1740,731) -- (1739,731) -- (1738,730) -- (1737,730);
		\fill[loptc] (1768,744) circle[];
		\draw[goptc] (1737,730) -- (1736,729) -- (1735,729) -- (1734,728) -- (1733,728) -- (1732,727) -- (1731,727) -- (1730,727) -- (1729,726) -- (1728,726) -- (1727,725) -- (1726,725) -- (1725,724) -- (1724,724) -- (1723,723) -- (1722,723) -- (1721,723) -- (1721,722) -- (1720,722) -- (1719,722) -- (1718,721) -- (1717,721) -- (1716,720) -- (1715,720) -- (1714,719) -- (1713,719) -- (1712,718) -- (1711,718) -- (1710,718) -- (1709,717) -- (1708,717) -- (1707,716) -- (1706,716) -- (1705,715) -- (1704,715) -- (1703,714) -- (1702,714) -- (1701,714) -- (1701,713) -- (1700,713) -- (1699,713) -- (1698,712) -- (1697,712) -- (1696,711) -- (1695,711) -- (1694,710) -- (1693,710) -- (1692,709) -- (1691,709) -- (1690,709) -- (1689,708) -- (1688,708) -- (1687,707) -- (1686,707) -- (1685,706) -- (1684,706) -- (1683,705) -- (1682,705) -- (1681,705) -- (1680,704) -- (1679,704) -- (1678,703) -- (1677,703) -- (1676,702) -- (1675,702) -- (1674,701) -- (1673,701) -- (1672,701) -- (1672,700) -- (1671,700) -- (1670,700) -- (1669,699) -- (1668,699) -- (1667,698) -- (1666,698) -- (1665,697) -- (1664,697);
		\fill[goptc] (1737,730) circle[];
		\draw[loptc] (1709,683) -- (1710,684) -- (1711,684) -- (1712,685) -- (1713,685) -- (1714,686) -- (1715,686) -- (1716,686) -- (1717,687) -- (1718,687) -- (1719,688) -- (1720,688) -- (1721,689) -- (1722,689) -- (1723,690) -- (1724,690) -- (1725,691) -- (1726,691) -- (1727,691) -- (1728,692) -- (1729,692) -- (1730,693) -- (1731,693) -- (1732,694) -- (1733,694) -- (1734,695) -- (1735,695) -- (1736,695) -- (1736,696) -- (1737,696) -- (1738,696);
		\fill[loptc] (1709,683) circle[];
		\draw[goptc] (1738,696) -- (1739,697) -- (1740,697) -- (1741,698) -- (1742,698) -- (1743,699) -- (1744,699) -- (1745,700) -- (1746,700) -- (1747,700) -- (1747,701) -- (1748,701) -- (1749,701) -- (1750,702) -- (1751,702) -- (1752,703) -- (1753,703) -- (1754,704) -- (1755,704) -- (1756,705) -- (1757,705) -- (1758,705) -- (1758,706) -- (1759,706) -- (1760,706) -- (1761,707) -- (1762,707) -- (1763,708) -- (1764,708) -- (1765,709) -- (1766,709) -- (1767,710) -- (1768,710) -- (1769,710) -- (1769,711) -- (1770,711) -- (1771,711) -- (1772,712) -- (1773,712) -- (1774,713) -- (1775,713) -- (1776,714) -- (1777,714) -- (1778,715) -- (1779,715) -- (1780,715) -- (1780,716) -- (1781,716) -- (1782,716) -- (1783,717) -- (1784,717) -- (1785,718) -- (1786,718) -- (1787,719) -- (1788,719) -- (1789,720) -- (1790,720) -- (1791,720) -- (1791,721) -- (1792,721) -- (1793,721) -- (1794,722) -- (1795,722) -- (1796,723) -- (1797,723) -- (1798,724) -- (1799,724) -- (1800,725) -- (1801,725) -- (1802,725) -- (1802,726) -- (1803,726) -- (1804,726) -- (1805,727) -- (1806,727) -- (1807,728) -- (1808,728) -- (1809,729) -- (1810,729) -- (1811,730) -- (1812,730) -- (1813,730);
		\fill[goptc] (1738,696) circle[];
		\end{tikzpicture}
	}
	\caption{For $n = 1664,1665,\ldots,1813$, the optimal values of $c$ are drawn in dark green. Other local optimal values of $c$ are drawn in amber. The graph displays $2$ tracks of local optimal values of $c$. The high-valued track ends at $n = 1768$. The low-valued track begins at $n = 1709$. The optimal value of $c$ for $n=1738$ is $34$ lower than the optimal value of $c$ for $n=1737$.} \label{fig:loptc}
\end{figure}
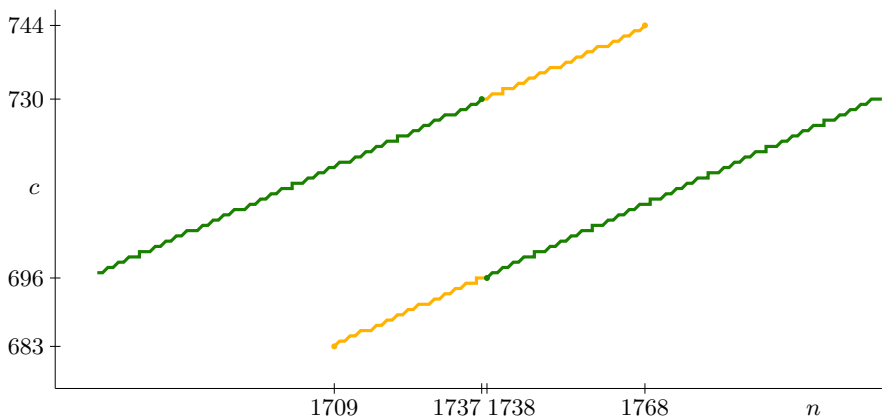
For larger $n$, there are intervals $[n_1,n_2]$ in the range of values of $n$, on which there are $2$ tracks of local optimal values of $c$. For $n = n_1$, the track with the highest value of $c$ gives the optimal value. This track is also the continuation of the optimal values of $c$ for $n < n_1$: the other track appears out of the blue. For $n = n_2$, the track with the lowest value of $c$ gives the optimal value. This track also continues with optimal values of $c$ for $n > n_2$: the other track disappears into thin air.  Figure \ref{fig:loptc} illustrates this phenomenon for $n_1 = 1709$ and $n_2 = 1768$. Somewhere between $n_1$ and $n_2$, the optimal value of $c$ switches from the high-valued track to the low-valued track and drops. 

The table on the next page indicates these drops. The table indicates the tracks of local optimal values of $c$ as follows. The first track starts at $n = 2$ and ends at $n = 47$, and is optimal all the time. The $(i+1)$\textsuperscript{th} track starts at $n=n_1$ in the $i$\textsuperscript{th} row and last but one column, and becomes optimal at the value of $n$ in the $i$\textsuperscript{th} row and last column.  The $(i+1)$\textsuperscript{th} track is optimal for the last time at  the value of $n$ in the $(i+1)$\textsuperscript{th} row and first column, and the track ends at $n=n_2$ in the $(i+1)$\textsuperscript{th} row and second column.
\begin{center}
	\begin{tikzpicture}[x=4.5mm,y=-5mm]
	\fill[tableIndex] (-11,0) rectangle (-9,8) (9,0) rectangle (11,8);
	\fill[tableBold] (-3,0) rectangle (3,8);
	\fill[tableFrame] (-11,-1.1) rectangle (11,0);
	\fill[tableVoid] (-1,-1.1) rectangle (1,0);
	\node at (0,-0.5) {\bf drop\strut};
	\begin{scope}[white,shift={(0,-0.5)}]
	\node at (-2,0) {{\mathversion{bold}$c'\mathstrut$}};
	\node at (2,0) {{\mathversion{bold}$c'\mathstrut$}};
	\node at (-5,-0.05) {{\mathversion{bold}$r(C^{c'}_n)$}};
	\node at (5,-0.05) {{\mathversion{bold}$r(C^{c'}_n)$}};
	\node at (-8,0) {{\mathversion{bold}$n_2\mathstrut$}};
	\node at (8,0) {{\mathversion{bold}$n_1\mathstrut$}};
	\node at (-10,0) {{\mathversion{bold}$n\mathstrut$}};
	\node at (10,0) {{\mathversion{bold}$n\mathstrut$}};
	\end{scope}
	\draw[white] \foreach \x in {-9,-7,-3,3,7,9} {(\x,-1.05) -- (\x,0)};
	\draw[tableFrame,very thick] \foreach \x in {-9,-7,-3,3,7,9} {(\x,0) -- (\x,8)};
	\draw[tableFrame,very thick] (-1,-1.1) -- (-1,8) (1,-1.1) -- (1,8);
	\foreach[count=\y] \na/\ba/\ca/\ra/\nb/\ab/\cb/\rb/\g in {47/47/15/3331/48/48/14/3490/1, 99/100/35/17323/99/98/33/17323/2, 204/207/78/84024/205/202/73/84936/5, 418/426/166/396403/419/412/157/398437/9, 854/869/350/1836388/855/840/333/1841006/17, 1737/1768/730/8347386/1738/1709/696/8357520/34, 3524/3583/1508/37445730/3525/3468/1444/37468248/64, 7132/7246/3097/166023725/7133/7024/2977/166072093/120} {
		\draw[tableFrame,very thick] (-11,\y-1) -- (11,\y-1);
		\node at (0,\y-0.44) {\bf\g};
		\node at (-2,\y-0.44) {\ca};
		\node at (+2,\y-0.44) {\cb};
		\node at (-5,\y-0.44) {\ra};
		\node at (+5,\y-0.44) {\rb};
		\node at (-8,\y-0.44) {\ba};
		\node at (+8,\y-0.44) {\ab};
		\node at (-10,\y-0.44) {\na};
		\node at (+10,\y-0.44) {\nb};
	}
	\draw[tableFrame,very thick] (-11,-1.1) rectangle (11,8);
	\end{tikzpicture}
\end{center}

Observe that (local) optimal values of $c$ can be double: $c + 1$ can also be a (local) optimal value. This is the case for $n = 13$, where both $c =2$ and $c = 3$ are optimal (see also Figure~\ref{fig:C13_2&3}). We did not find triple (local) optimal  values, though. For $n = 3512$, both the optimal value of $c$ and the other local optimal value of $c$ are double: 
\begin{align*}
r(C^{1438}_{3512}) &= 37170635 = r(C^{1439}_{3512}) &
r(C^{1502}_{3512}) &= 37180596 = r(C^{1503}_{3512}).
\end{align*}
This is the only value of $n \le 10000$ for which this occurs.
\fi

\section{An improvement of Martyugin's prime number construction of binary PFAs} \label{sec:Primes}

Although the \v{C}ern\'y family contains extremal binary PFAs for all $n\leq 10$, it only gives polynomial reset thresholds for large $n$. In this section and the next section, we show that for $n\geq 41$, the \v{C}ern\'y family is not extremal anymore. We do this by presenting a construction based on prime numbers.  

Our construction is an improvement of the binary prime number construction by Pavel Martyugin in \cite{Mar08}, see also \S 6 of \cite{B18}. Martyugin's construction uses $1+2\sum_{i=1}^r p_i$ states and has reset threshold $2\prod_{i=1}^r p_i$, where $p_1,\ldots,p_r$ are the first $r$ primes, which is improved to $2+2\prod_{i=1}^r p_i$ in \cite{B18}. Martyugin has a ternary construction as well, with $1+\sum_{i=1}^r p_i$ states and reset threshold $1+\prod_{i=1}^r p_i$.

Let the binary PFA with $1+2\sum_{i=1}^r p_i$ states and reset threshold $2+2\prod_{i=1}^r p_i$ be called $M^r$. For instance, taking $r=5$ gives $p_1,\ldots,p_5 = 2,3,5,7,11$ and yields a PFA $M^5$ with $n=57$ states and reset threshold 4622. This is not yet sufficient to overtake the \v{Cern\'y} family, since $r(C_{57}^{18}) = 5152 > 4622$. 

Martyugin's construction can be generalized in a very easy manner, because we do not need to restrict ourselves to the use of the first $r$ primes. One can use any list $\mathbf{p}= (p_1,p_2,\ldots,p_r)$ of relatively prime numbers, and the construction still works. Denote the corresponding PFA as $M^{\mathbf{p}}$ and let $q:= q(\mathbf{p}) = \prod_{i=1}^rp_i$. Compared to $M^r$, $M^{\mathbf{p}}$ offers more flexibility in the choice of numbers, leading to constructions for more state sets. But larger reset thresholds are possible as well. For instance, if $\mathbf{p} = (2,3,5,7,11,13,17,19,23)$, then $M^9 = M^{\mathbf{p}}$ has $201$ states, and reset threshold $2+2q = 446185742$. But if $\mathbf{p} = (5,7,9,11,13,16,17,19)$, then $M^{\mathbf{p}}$ only has $195$ states, and reset threshold $2+2q = 465585122$.

We present a further improvement, based on a list $\mathbf{p}= (p_1,p_2,\ldots,p_r)$ of relatively prime numbers as well. The construction of our binary PFA $P^{\bf p}$ is illustrated in Figure \ref{prime41} for $r = 4$ and $p_1 = 5$, $p_2 = 7$, $p_3 = 8$ and $p_4 = 9$.
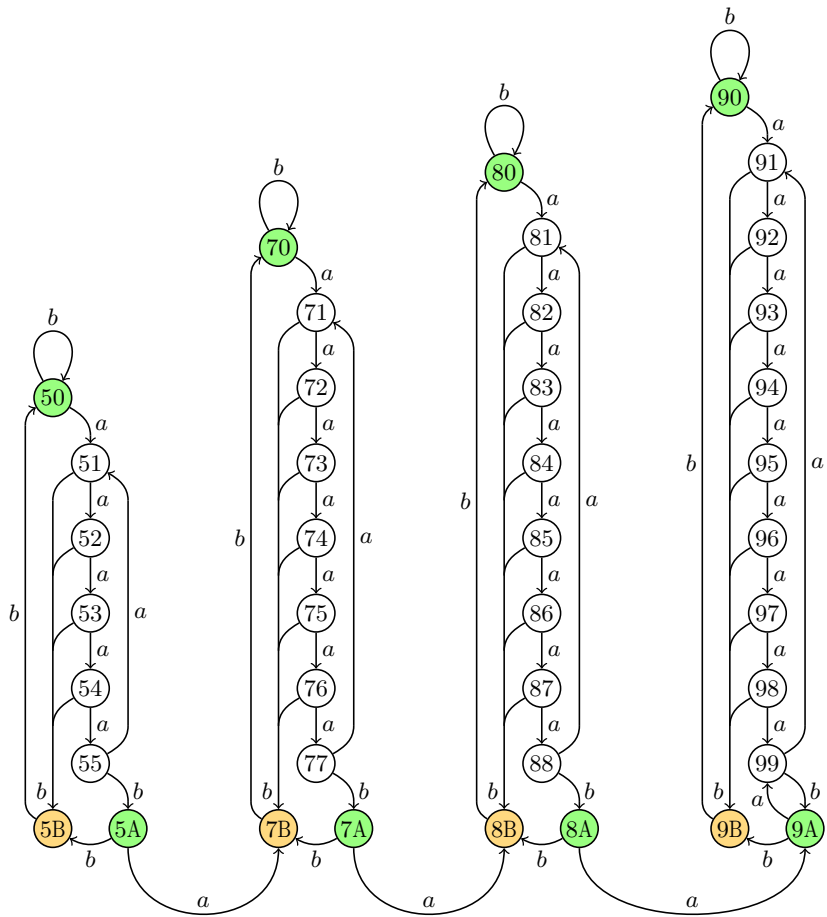
\begin{figure}[!h!t]
\centering{
\begin{tikzpicture}[semithick]
\pgfmathsetmacro{\groupdist}{2}
\foreach \p [count=\n] in {5, 7, 8, 9} {
\begin{scope}[shift={(\n+\groupdist*\n,0)}]
\foreach \i in {1,...,\p} {
	\node[mynode,draw,inner sep=0pt] (\p\i) at (0,\p-\i) {\small\p\i};
}
\node[mynode,draw,fill=CernyColC] (\p T) at (-60:1) {\small\p\scalebox{0.8}[1]{A}};
\node[mynode,draw,fill=CernyColB] (\p L) at (-120:1) {\small\p\scalebox{0.8}[1]{B}};
\node[mynode,draw,fill=CernyColC,shift={(120:1)}] (\p S) at (\p 1) {\small\p{0}};
\pgfmathsetmacro{\pmm}{\p-1} 
\foreach \i in {1,...,\pmm} {
	\node[mynode,shift={(0,-1)}] (tipp) at (\p\i) {};
      \draw[->] (\p\i) -- (tipp);
      \node[anchor=west,shift={(0,-0.5)},inner sep=2pt] at (\p\i) {$a$};
	\draw (\p\i) edge[out=-150,in=90] ++(-0.5,-0.5);
}
\draw[->] (\p 1) + (-0.5,-0.5) -- (\p L);
\node[anchor=east,shift={(0,0.5)},inner sep=2pt] at (\p L) {$b$};
\draw (\p\p) edge[out=30,in=-90] ++(0.5,0.5) ++(0.5,0.5) edge ++(0,\p-2);
\draw (\p 1) + (0.5,-0.5) edge[->,out=90,in=-30] (\p 1);
\node[anchor=west,inner sep=2pt] at (0.5,0.5*\p-0.5) {$a$};
\draw (\p\p) edge[out=-30,in=90] ++(0.5,-0.5);
\draw[->] (\p\p) + (0.5,-0.5) -- (\p T);
\node[anchor=west,shift={(0,0.5)},inner sep=1.5pt] at (\p T) {$b$};
\draw (\p T) edge[->,out=-150,in=-30] node[below,inner sep=2pt] {$b$} (\p L);
\draw (\p L) edge[out=150,in=-90] ++(-0.366,0.366) ++(-0.366,0.366) edge ++(0,\p);
\draw (\p S) + (-0.366,-0.366) edge[->,out=90,in=-150] (\p S);
\node[anchor=east,inner sep=2pt] at (-0.866,0.5*\p-0.5) {$b$};
\draw (\p S) edge[out=-30,in=90] ++(0.5,-0.5);
\draw[->] (\p S) + (0.5,-0.5) -- (\p 1);
\node[anchor=west,shift={(0,0.5)},inner sep=1.5pt] at (\p 1) {$a$};
\draw (\p S) edge[->,out=120,in=60,looseness=10] node[above,inner sep=2pt] {$b$} (\p S);
\end{scope}
}
\pgfmathsetmacro{\lons}{(\groupdist+1)/\groupdist}
\draw (5T) edge [->,out=-90,in=-90,looseness=\lons] node[above,inner sep=2.5pt] {$a$} (7L);
\draw (7T) edge [->,out=-90,in=-90,looseness=\lons] node[above,inner sep=2.5pt] {$a$}(8L);
\draw (8T) edge [->,out=-90,in=-90,looseness=1] node[above,inner sep=2.5pt] {$a$} (9T);
\draw (9T) edge[out=150,in=-90] ++(-0.5,0.5);
\draw[->] (9T) + (-0.5,0.5) -- (99);
\node[anchor=east,shift={(0,-0.5)},inner sep=1pt] at (99) {$a$};
\end{tikzpicture}
}
\caption{A binary PFA with $41$ states which takes $3114$ steps to synchronize, which is more than any PFA of the \v{C}ern\'y family with $41$ states.} \label{prime41}
\end{figure}

The state set $Q$ of the PFA $P^{\mathbf p}$ consists of $r$ groups. The $i$\textsuperscript{th} group contains $p_i + 3$ states, namely
$$
p_i 0,p_i 1,p_i 2,\ldots,p_i p_i,p_i\mathrm{A},p_i\mathrm{B}
$$
The transitions are indicated in Figure \ref{prime41}. Groups $i$ and $i+1$ are connected by a transition of symbol $a$ for each $i$. Notice that the connection for $i = r-1$ differs from those for $i < r-1$. This makes the reset threshold a little larger.

Symbol $a$ is undefined on state $p_i\mathrm{B}$ for each $i$. For that reason, any synchronizing word starts with $b^3$. The word $b^3$ resets group $i$ to state $p_i 0$ for each $i$, so $Qb^3 = \{p_1 0,p_2 0, \ldots, p_r 0\}$. Furthermore, $Q b^3 a^q = \{p_1 p_1,p_2 p_2, \ldots, p_r p_r\}$ and $Q b^3 a^q b = \{p_1\mathrm{A},p_2\mathrm{A}, \ldots, p_r\mathrm{A}\}$. We infer that $Q b^3 a^q b a$ is defined and that it does not contain states of group $1$. By induction on $r$, it follows that the PFA $P^{\mathbf{p}}$ is synchronizing.

Suppose that $w$ is a shortest synchronizing word which starts with $b^3 a^j b$ for some $j < q$. If $j = 0$, then $Qb^3 a^j b = Qb^3$, which contradicts the minimality of $|w|$. So $j > 0$. Since $0 < j < q$, it follows that there exists an $i$ such that $(p_i 0) a^j \ne (p_i p_i)$. Hence $(p_i 0) a^j b = (p_i \mathrm{B})$. From this, one can infer that either $Qb^3 a^j b^2 = Qb^3$ or $Qb^3 a^j b^3 = Qb^3$, and that $|w|$ is not minimal, which is a contradiction.
So $w$ starts with $b^3 a^q$, and $|w| = \Omega(q)$. It is not hard to see that $|w| = O(q)$, so $|w| = \Theta(q)$. 

It is a little harder to find the exact reset threshold given in the next proposition. 

\begin{proposition}\label{prop:primereset} If the PFA $P^{\mathbf{p}}$ is constructed with relative primes $p_1,\ldots,p_r$, then it has $3r+\sum_{i=1}^r p_i$ states and reset threshold
	$$
	r(P^{\mathbf{p}}) = 5r - 2+ \sum_{i=1}^{r-1} 
	p_i p_{i+1} \cdots p_r .
	$$
\end{proposition}
The proof (by induction) is left as an exercise to the interested reader. 

Let $P^{r}$ be the PFA $P^{\mathbf{p}}$ with $\mathbf{p}$ the list of the first $r$ primes (in increasing order). The following proposition shows that $M^r$ is defeated by $P^{r+1}$ if $r \ge 6$, and that $M^{\mathbf{p}}$ is defeated if it has at least $62$ states.


\begin{proposition} \mbox{}
	\begin{enumerate}[(i)]
		\item Suppose that $r\geq 6$. Then $P^{r+1}$ has fewer states and larger reset threshold than $M^r$.
		\item Suppose that $M^{\mathbf{p}}$ has at least 62 states. Then there exists a prime power $p\geq 5$  with the following properties: the list $\mathbf{p}$ can be extended with one element by inserting $p$ at any position to obtain $\mathbf{p'}$, and
$P^{\mathbf{p'}}$ has fewer states and larger reset threshold than $M^{\mathbf{p}}$.
		\end{enumerate}
\end{proposition}
	
%

\begin{proof}
Let $n(A)$ denote the number of states of a PFA $A$. 
First observe that $r(P^{r+1}) \geq p_{r+1}\cdot \prod_{i=1}^r p_i > 2+2\cdot \prod_{i=1}^r p_i = r(M^r)$, where $p_i$ is the $i$\textsuperscript{th} prime for each $i$. Similarly, $r(P^{\mathbf{p'}})\geq 5\cdot\prod_{p\in \mathbf{p}} p>r(M^{\textbf{p}})$.	So it remains to show that $n(P^{r+1}) < n(M^r)$ and $n(P^{\mathbf{p'}}) < n(M^{\mathbf{p}})$.
\begin{enumerate}[(i)]

\item 
We prove $n(P^{r+1}) < n(M^r)$ by induction on $r$. The PFA $M^6$ has $2\cdot(2+3+5+7+11+13)+1 = 83$ states, and $P^7$ has $2+3+5+7+11+13+17+7\cdot3 = 79$ states, so the case $r=6$ is satisfied. Suppose that $r \ge 7$ and that $P^r$ has fewer states than $M^{r-1}$. Then $M^r$ has $2p_r$ more states than $M^{r-1}$, and $P^{r+1}$ has $p_{r+1}+3$ more states than $P^r$. From Bertrand's postulate, it follows that $2p_r \ge p_{r+1}+3$, so $n(P^{r+1}) < n(M^r)$.
 
\item 
\ifijfcs
The proof of $n(P^{\mathbf{p'}}) < n(M^{\mathbf{p}})$ can be found in the arXiv version of this paper \cite{arXiv}.
\else
Let $r$ be the number of elements in $\mathbf{p}$ and let $t$ be the number of distinct prime factors in these elements. We prove $n(P^{\mathbf{p'}}) < n(M^{\mathbf{p}})$ by distinguishing the cases $t \le 5$ and $t \ge 6$.

Suppose first that $t \ge 6$. Since $r \le t$ and $\sum_{p\in\mathbf{p}} p$ is at least the sum of the first $t$ prime numbers, we infer that
$$
n\big(M^{\mathbf{p}}\big) - n\big(P^{\mathbf{p}}\big) = 1-3r+{\textstyle \sum_{p\in\mathbf{p}} p} \ge n\big(M^t\big) - n\big(P^t\big).
$$
Since the $(t+1)$\textsuperscript{th} prime exceeds both $8$ and $9$ (the $7$\textsuperscript{th} prime is $17$), we can construct $\mathbf{p'}$ by choosing a prime power $p\ge 5$ which is at most the $(t+1)$\textsuperscript{th} prime. So
$$
n\big(P^{\mathbf{p'}}\big) - n\big(P^{\mathbf{p}}\big) \le
n\big(P^{t+1}\big) - n\big(P^t\big).
$$
Consequently,
$$
n\big(M^{\mathbf{p}}\big) - n\big(P^{\mathbf{p'}}\big)  \ge
n\big(M^t\big) - n\big(P^{t+1}\big).
$$
So the case $t \ge 6$ follows from claim (i).

Suppose next that $t \le 5$. Assume that $n(M^\mathbf{p}) \ge 62$. Then $r \le t \le 5$ and
$$
n(P^\mathbf{p})  = \tfrac12 \big(n(M^\mathbf{p})-1\big)+3r < n(M^\mathbf{p}) -31+15.
$$
Construct $\mathbf{p'}$ by taking for $p$ the smallest power $\ge 5$ of the smallest prime number that does not yet occur as a factor in $\mathbf{p}$. As $t \le 5$, $p \in \{8,9,5,7,11,13\}$ follows. So $p + 3 \le 31 - 15$ and
$$
n\big(P^{\mathbf{p'}}\big) = n\big(P^{\mathbf{p}}\big) + p + 3 < n(M^\mathbf{p}),
$$
which completes the proof of claim (ii).
\fi
\qed
\end{enumerate}
\end{proof}

Just like for $M^{\mathbf{p}}$ and $M^r$, $P^{\mathbf{p}}$ offers more flexibility and sometimes a better reset threshold than $P^r$. For instance, $P^4$ has $29$ states and $r(P^4) = 368$, and $P^5$ has $43$ states and $r(P^5) = 3950$. So $P^r$ does not exist with $41$ states, and adding additional states to $P^4$ to fix that yields a poor result. 
But if $\mathbf{p} = (5,7,8,9)$, then $P^{\mathbf{p}}$ has 41 states as well, and $r(P^{\mathbf{p}}) = 3114$, see Figure \ref{prime41}. The following is an example where $P^{\mathbf{p}}$ has a better reset threshold than $P^r$. If $\mathbf{p} = (2,3,5,7,11,13,17)$, then $P^7 = P^{\mathbf{p}}$ has $79$ states, and $r(P^7) = r(P^{\mathbf{p}}) = 870552$. But if $\mathbf{p} = (5,7,9,11,13,16)$, then $P^{\mathbf{p}}$ has the same number of states, and $r(P^{\mathbf{p}}) = 887980$. The order of the primes in $\mathbf{p}$ is relevant to obtain a value which exceeds $870552$.

The construction $P^{\mathbf{p}}$ (in particular $P^r$) is not transitive, but we can make it transitive with the following modification. We change $(p_i 0) a$ for all $i$ from $(p_i 0) a = (p_i 1)$ to
\begin{align*}
(p_1 0) a &= (p_r 1) & (p_i 0) a &= (p_{i-1} 1) \qquad (2 \le i \le r)
\end{align*}
One can verify that this modification gives a transitive PFA $\tilde P^{\mathbf{p}}$, which still satisfies $r(\tilde P^{\mathbf{p}}) = \Theta(q)$. It is not hard to prove that $r(\tilde P^{\mathbf{p}}) \geq q+q/{\max\{p_1,p_2,\ldots,p_r\}}$.

If we want to have extra states as well (to obtain a specific $n$), then we add them in such a way that transitivity is preserved. We add such states between state $p_r \mathrm{B}$ and state $p_r 0$, as follows.
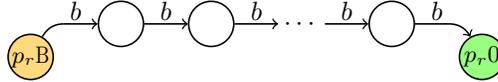
\begin{center}
\begin{tikzpicture}[x=12mm,y=12mm,semithick]
\node[mynode,minimum width=6mm,draw,fill=CernyColB] (9L) at (0,-0.366) {\small $p_r$\scalebox{0.8}[1]{B}};
\node[mynode,minimum width=6mm,draw] (E1) at (1,0) {};
\node[mynode,minimum width=6mm,draw] (E2) at (2,0) {};
\node[mynode,minimum width=6mm] (E3) at (3,0) {$\cdots$};
\node[mynode,minimum width=6mm,draw] (E4) at (4,0) {};
\node[mynode,minimum width=6mm,draw,fill=CernyColC] (9S) at (5,-0.366) {\small $p_r$0};
\draw (9L) edge[out=60,in=-180] ++(0.366,0.366) ++(0.366,0.366) edge[->] (E1);
\draw[->] (E1) edge (E2) (E2) edge (E3) (E3) edge (E4);
\draw (E4) -- ++ (0.634,0) edge[->,out=0,in=120] (9S);
\draw foreach \n in {0,...,4} {(\n+0.5,0) node[anchor=south,inner sep=1pt] {$b$}};
\end{tikzpicture}
\end{center}
So symbol $a$ is undefined on the extra states. Adding additional states this way makes the reset threshold larger, but only marginally. So $|w| = \Theta(q)$ is not affected. 

We can use $|w| = \Theta(q)$ to derive an asymptotic estimate for both Martyugin's  construction and our construction. This has already been done for $M^r$ and its ternary variant in \cite{B18}, and $P^r$ can be settled in a similar manner. But the estimate in \cite{B18} does not distinguish between $M^r$ and $P^r$. The estimates for $r(M^r)$, its ternary variant, and $r(P^r)$, as given in \cite{B18}, are
\[
\exp\big({\Theta(1)\cdot\sqrt{n\cdot\ln(n)}}\big).
\]
The upper bound is valid for $r(M^{\mathbf{p}})$ and $r(P^{\mathbf{p}})$ as well. For the lower bound, an extra condition is required, since one can make very poor constructions, typically with $r = 1$, where $r$ is the number of elements of $\mathbf{p}$. A condition which suffices is $n = O\big(r^2 \log(r)\big)$. 
Notice that the $r$\textsuperscript{th} prime is $\Theta\big(r\log(r)\big)$. So $n = O\big(r^2 \log(r)\big)$ holds if $p_i$ is at most proportional to the $r$\textsuperscript{th} prime for each $i$. 

The above estimates were derived to show that Martyugin's construction was strictly between polynomial and exponential. But to compare Martyugin's construction and our construction, we need better estimates. To derive such estimates, we will use some classical results of number theory. After that, one can easily conclude that the reset threshold of our construction is better than Martyugin's construction, by a factor which lies strictly between polynomial and exponential.


\begin{theorem}Let $n\in \mathbb{N}$ and suppose that $r$ is maximal such that $n(M^r)\leq n$. Let $\mathbf{p}$ be such that $n(M^{\mathbf{p}})\leq n$. Then
\[
	r(M^r) = \exp\big({(1\pm o(1))\sqrt{n\cdot\ln(n)/2}}\big),\qquad r(M^{\mathbf{p}}) \leq \exp\big({(1+ o(1))\sqrt{n\cdot\ln(n)/2}}\big).
\]
Let $n\in \mathbb{N}$ and suppose that $r$ is maximal such that $n(P^r)\leq n$. Let $\mathbf{p}$ be such that $n(P^{\mathbf{p}})\leq n$. Then
\[
r(P^r) = \exp\big({(1\pm o(1))\sqrt{n\cdot\ln(n)}}\big),\qquad r(P^{\mathbf{p}}) \leq \exp\big({(1+ o(1))\sqrt{n\cdot\ln(n)}}\big).
\]
\end{theorem}

\begin{proof}
We only prove the second claim, because the proof of the first claim is similar. Using the asymtotic estimate $\sum_{i=1}^r p_i \sim \tfrac12 r^2 \ln(r)$, we infer that
$$
n \sim n(P^r) \sim \tfrac12 r^2 \ln(r) + O(r) = r^{2 + o(1)}.
$$
Consequently,
$$
n \sim \tfrac12 r^2 \ln(r) = \tfrac12 r^2 \frac{\ln(n)}{2 + o(1)} \sim \tfrac14 r^2 \ln(n),
$$
and
$$
r \sim \sqrt{4n /{\ln(n)}} = n ^ {(1 - o(1))/2}.
$$
Using the asymtotic estimate $\prod_{i=1}^r p_i \sim r^{(1\pm o(1))r}$ and $r(P^r) = \Theta\big(\prod_{i=1}^r p_i\big)$, 
we obtain
$$
r(P^r) \sim r^{(1 \pm o(1)) r} = n^{(1\pm o(1))r/2}. 
$$
This yields the estimate for $r(P^{r})$.

We proceed with the estimate for $r(P^{\mathbf{p}})$. Notice that $\mathbf{p}$ contains $r'\leq r$ relatively prime numbers. So $\prod_{p\in\mathbf{p}}p$ is the product of $r'$ numbers whose sum is less than $n$. This product is at most $(n/r')^{r'}$. Since $n/r' \ge n/r \ge \exp(1)$ for large $n$, we infer that $(n/r')^{r'} \le (n/r)^r$ for large $n$. So by Proposition \ref{prop:primereset},
$$
r(P^{\mathbf{p}}) \le r\cdot {\textstyle \prod_{p\in\mathbf{p}}p} \leq r (n/r)^r = r\big(n^{(1+o(1))/2}\big)^r = n^{(1+ o(1))r/2},
$$
which yields the estimate for $r(P^{\mathbf{p}})$.
\end{proof}

Martyugin's ternary construction satisfies the same estimate as our binary construction.

\section{Binary PFAs with larger reset thresholds for $n \ge 41$ states}\label{sec:LargerRT}

Now that we have the improved prime number construction, it remains to prove that $\mathbf{p}$ can be chosen such that $P^{\mathbf{p}}$ and $\tilde P^{\mathbf{p}}$ defeat the \v{C}ern\'y family for all $n \ge 41$. In order to do this, we will construct a series of PFAs of this type, for which $n = O\big(r^2 \log(r)\big)$ does \emph{not} hold. Instead, we construct PFAs of which the reset threshold is only $\Theta(n^3)$. This is sufficient and more easy.

To show for any construction that it defeats the \v{C}ern\'y family for large enough $n$, we first need an upper bound for the reset thresholds of the automata in the \v{C}ern\'y family which is valid for all $c$, but does not depend on $c$. Again, let $n' = n - c - 1$.
From Theorem \ref{thm:reduction} and Proposition \ref{prop:bounds_fc}, it follows that
\begin{align*}
r(C^c_n) &\le n'(n'-1) + c+1 + f_c(n')\\
 &\le (n')^2 + (c+1)^2+ (c+1)n' \big\lceil \log_2(n) \big\rceil\\
 &= (n'+c+1)^2+  (c+1)n'\big\lceil \log_2(n)-2 \big\rceil\\
 &\le \tfrac14 n^2 \big\lceil \log_2(n)+2 \big\rceil.
\end{align*}
This is not really a good upper bound, but it will be sufficient for all $n \ge 47$. For 
$41 \le n \le 46$, our general argument will not be precise enough.

We first assume that $n \ge 47$. Notice that $\big\lceil \log_2(52) + 2 \big\rceil = 8$, and that 
$$
\big\lceil \log_2(m+52) + 2 \big\rceil \le \tfrac1{10}m + 8 \le \tfrac29(m+36)
$$ 
%
for all $m \in \N$. Now take $m := \max\{n-52,0\}$. Then it follows from the above that
\begin{equation} \label{rest52}
\begin{split}
54\,r(C^c_{n}) &\le 3(m+36)(m+52)(m+52) \\
&\le 3(m+37)(m+51)(m+52) = (m+51)(m+52)(3m+111).
\end{split}
\end{equation}

Let $i \in \N$. The PFAs $P^{\mathbf{p}}$ and $\tilde{P}^{\mathbf{p}}$ of the $4$ relatively prime numbers $8$, $7+2i$, $9+2i$, $11+2i$ have
$$
(3 + 8) + (3 + 7 + 2i) + (3 + 9 + 2i) + (3 + 11 + 2i) = 47 + 6i
$$
states. For all $47 + 6i \le n < 53 + 6i$, we can use the same $4$ relatively prime numbers in our constructions of $P^{\mathbf{p}}$ and $\tilde{P}^{\mathbf{p}}$ with $n \ge 47$ states. From $i \ge 0$ and $6i \ge n - 52$, we infer that $2i \ge \frac{m}3$. So $q \ge 8\big(\tfrac{m}3 + 7\big)\big(\tfrac{m}3 + 9\big)\big(\tfrac{m}3 + 11\big)$ and
\begin{equation}
\begin{split}
\ifijfcs
\else
\!
\fi
54 q &\ge 16(m+21)(m+27)(m+33) \\
&\ge 16(m+21)(m+26)(m+34) 
= \big(\tfrac{17}{7}m + 51\big)(2m+52)\big(\tfrac{56}{17}m + 112\big).
\end{split}
\end{equation}
By way of a factor comparison with \eqref{rest52}, we see that $P^{\mathbf{p}}$ and $\tilde{P}^{\mathbf{p}}$ defeat the \v{C}ern\'y family if the number of states is at least $47$.  

With $41 \le n \le 46$ states, we can choose specific relatively prime numbers, and compare the reset threshold of the best PFA of the \v{C}ern\'y family, say $C_n^{c'}$, with the product $q$. The columns with r.t.\@ and r.t.~t.\@ give the reset thresholds for $P^{\mathbf{p}}$ and the transitive variant $\tilde P^{\mathbf{p}}$ respectively.
\begin{center}
\begin{tikzpicture}[x=2mm,y=-5mm]
\fill[tableIndex] (0,0) rectangle (4,6) (41,0) rectangle (45,6);
\fill[tableBold] (8,0) rectangle (20,6);
\fill[tableFrame] (0,-1.1) rectangle (45,0);
\fill[red] (14,1) rectangle (20,2);
\begin{scope}[white,shift={(0,-0.5)}]
\node at (2,0) {{\mathversion{bold}$n\mathstrut$}};
\node at (6,0) {{\mathversion{bold}$c'\mathstrut$}};
\node at (11,-0.05) {{\mathversion{bold}$r(C^{c'}_n)$}};
\node at (17,0) {{\mathversion{bold}$q\mathstrut$}};
\node at (23,0) {{\bf r.t.\strut}};
\node at (29,0) {{\bf r.t.~t.\strut}};
\node at (36.5,0) {{\bf primes\strut}};
\node at (43,0) {{\mathversion{bold}$n\mathstrut$}};
\end{scope}
\draw[white] \foreach \x in {4,8,14,20,26,32,41} {(\x,-1.05) -- (\x,0)};
\draw[tableFrame,very thick] \foreach \x in {4,8,14,20,26,32,41} {(\x,0) -- (\x,6)};
\foreach[count=\y] \n/\c/\rc/\q/\rp/\rpt\prs in {41/13/2465/2520/3114/3056/{(5,7,8,9)}, 42/13/2601/2520/3117/3062/{(5,7,8,9)}, 43/13/2739/3080/3802/3726/{(5,7,8,11)}, 44/14/2882/3465/4275/4177/{(5,7,9,11)}, 45/14/3028/3960/4869/4683/{(5,8,9,11)}, 46/15/3177/3960/4872/4689/{(5,8,9,11)}}
{
\draw[tableFrame,very thick] (0,\y-1) -- (45,\y-1);
\node at (2,\y-0.44) {\n};
\node at (6,\y-0.44) {\c};
\node at (11,\y-0.44) {\bf\rc};
\node at (17,\y-0.44) {\bf\q};
\node at (23,\y-0.44) {\rp};
\node at (29,\y-0.44) {\rpt};
\node at (36.5,\y-0.4) {\prs};
\node at (43,\y-0.44) {\n};
}
\draw[tableFrame,very thick] (0,-1.1) rectangle (45,6);
\end{tikzpicture}
\end{center}
It appears that the value of $q$ is insufficient for estimation if $n = 42$. But both constructions $P^{\mathbf{p}}$ and $\tilde P^{\mathbf{p}}$ require at least $5\cdot 7\cdot 8\cdot 9+5\cdot 7\cdot 8 = 2800$ steps to synchronize if $n=42$, which is sufficient for estimation.

\section{Conclusion}

The \v{C}ern\'y family presented in this paper contains for all $n\leq 10$ a binary PFA with $n$ states and maximal possible reset threshold.
The analysis for the different members of the family has been done in general, by determining the maximal reset threshold in terms of recurrent sequences.

We also have shown that for $n \ge 41$ the \v{C}ern\'y family does not contain extremal PFAs anymore. This is proved by a new prime number construction which outperforms earlier known constructions.
For large $n$, there are constructions based on rewrite systems as introduced in \cite{BDZ19} with exponentially large reset thresholds, but they are insufficient to beat the \v{C}ern\'y family for all $n\geq 41$.	

We leave it as an open question if the \v{C}ern\'y family contains any extremal PFA for $11\leq n\leq 40$. The largest reset thresholds of the \v{C}ern\'y family are given below.
\begin{center}	
\begin{tikzpicture}[x=10mm,y=-5mm]
\foreach[count=\t] \L in {{11/119, 12/146, 13/176, 14/211, 15/248, 16/288, 17/332, 18/379, 19/429, 20/483}, {21/539, 22/599, 23/663, 24/732, 25/804, 26/881, 27/961, 28/1044, 29/1132, 30/1222}, {31/1317, 32/1416, 33/1517, 34/1624, 35/1733, 36/1846, 37/1963, 38/2082, 39/2207, 40/2334}} {
   \begin{scope}[shift={(0,3*\t)}]
	\fill[tableIndex] (0,-1) rectangle (10,0);
	\fill[tableFrame] (-1.37,-1) rectangle (0,1);
	\draw[white] (-1.37,0) -- (0,0);
	\node[white,anchor=east] at (0,-0.44) {{\mathversion{bold}$n\,\mathstrut$}};
	
	\foreach[count=\i] \n/\p in \L {
		\draw[tableFrame,very thick] (\i-1,-1) -- (\i-1,1);
		\node at (\i-0.5,-0.44) {$\n\mathstrut$};
		\node at (\i-0.5,0.56) {$\p\mathstrut$};
	}
	\draw[tableFrame,very thick] (-1.37,-1) -- (10,-1) -- (10,1) -- (-1.37,1) -- cycle (0,0) -- (10,0);
	\node[white,anchor=east] at (0,0.53) {{\mathversion{bold}$r(C^{c'}_n)\mathstrut$}};
   \end{scope}
} 
\end{tikzpicture}
\end{center}

\section*{Acknowledgements}

We thank the referees for their positive remarks and suggestions. We also thank Hans Zantema for discussions on the topic.

\ifijfcs
\bibliographystyle{ws-ijfcs}
\else
\bibliographystyle{splncs04}
\fi	
	\bibliography{Cernyref}
\end{document}